\documentclass{article}

\usepackage{iclr2027_conference,times}

\usepackage{amsmath,amsfonts,bm}

\def\eqref#1{equation~\ref{#1}}

\def\1{\bm{1}}

\DeclareMathAlphabet{\mathsfit}{\encodingdefault}{\sfdefault}{m}{sl}
\SetMathAlphabet{\mathsfit}{bold}{\encodingdefault}{\sfdefault}{bx}{n}

\usepackage{amsmath,amsfonts,amssymb, amsthm}
\newtheorem{proposition}{Proposition}
\usepackage{braket}
\usepackage{graphicx}
\usepackage{booktabs}
\usepackage{multirow}
\usepackage{hyperref}
\usepackage{url}

\iclrfinalcopy 

\title{Prototype-Guided Optimization of \\ Quantum Data Embeddings with Linear Cost}

\author{
Seokhoon Jeong \\
Department of Statistics and Data Science \\
Yonsei University \\
Republic of Korea \\
\texttt{detra0504@yonsei.ac.kr}
\And
Daniel K. Park \\
Department of Statistics and Data Science \\
Department of Applied Statistics \\
Department of Quantum Information \\
Yonsei University \\
Republic of Korea \\
\texttt{dkd.park@yonsei.ac.kr}
}

\begin{document}

\maketitle

\begin{abstract}
    In supervised quantum machine learning with classical data, classification performance critically depends on the geometry induced by quantum data embedding. 
    Existing data-adaptive approaches rely on pairwise fidelity evaluations between samples, which scale as $O(B^2)$ for $B$ samples and create a computational bottleneck in iterative embedding optimization.
    To address this issue, we propose a prototype-guided framework for efficient data-adaptive quantum embedding.    
    Using class-prototype separation and sample-to-prototype compactness, we derive a measurable lower certificate for class-ensemble distinguishability and construct a prototype-based objective requiring only $O(B)$ fidelity evaluations. 
    We integrate this objective into an end-to-end pipeline that combines task-specific quantum circuit architecture optimization with continuous refinement of the data-embedding parameters.
    Across five benchmark datasets with up to 20 qubits, our method retains competitive quantum kernel support vector machine (SVM) performance while reducing fidelity-evaluation cost, with particularly strong gains over budget-matched pairwise optimization on T-MNIST.
    Geometric analyses examine how the learned embeddings relate to kernel alignment and classification margins, revealing dataset-dependent benefits and limitations.
    We further derive a prototype-based classifier in the reproducing kernel Hilbert space that provides an explicit witness for the SVM objective and a low-cost proxy for model selection, and extend the framework to multiple prototypes for multimodal class structures.
\end{abstract}

\section{Introduction}

Quantum machine learning (QML) seeks to harness quantum information processing to enhance how data are represented, processed, and learned from \citep{qml_biamonte,qml_intro}. When classical data are used in a quantum learning model, they must first be encoded into quantum states, making the choice of quantum data embedding a fundamental design decision~\citep{q_featuremap,NQE}. The embedding determines how relationships in the input space are represented in quantum state space and can therefore strongly influence the behavior of downstream learning models~\citep{q_featuremap,powerofdata,NQE}. 

This role is particularly evident in supervised QML, where the objective is to learn a predictive relation between input data and their labels. In this setting, the geometry of the embedded states influences how effectively different classes can be distinguished~\citep{TQE,NQE,quantum_margin}. This dependence appears across different quantum classification models. In quantum kernel methods, fidelities between embedded quantum states define the kernel used by a classical learner such as a support vector machine (SVM)~\citep{q_featuremap,zzfeaturemap}, making the induced geometry directly relevant to kernel alignment and predictive performance~\citep{powerofdata,qk_inductive,covariant_qkernel,TQE_physicalreview}. In variational quantum classifiers, the embedding likewise affects expressive power~\citep{encoding_expressive,data_reupload} and predictive performance~\citep{NQE,GQE}, trainability~\citep{NQE, qml_trainability}, and generalization~\citep{quantum_margin,PQC_encoding_gen,qml_generalization_banchi}. Designing task-adaptive quantum embeddings is therefore a central problem in supervised QML.

The importance of quantum data embedding in supervised QML has motivated the development of various data-adaptive embedding approaches.
Existing methods construct embeddings by learning classical transformations of the input values~\citep{NQE}, optimizing trainable parameters within a predefined circuit ansatz~\citep{TQE,TQE_physicalreview}, or extending the optimization to the circuit architecture and its parameterization~\citep{GQE, IOP_NSGA2, GA_moretti,graph_search,neuro_serach_iclr}. Despite their methodological differences, these approaches share the goal of inducing a quantum representation that is better aligned with the target learning task.

Many approaches characterize this task through the pairwise geometry of embedded samples. Whether assessed through supervised learning objectives, kernel--target alignment, or downstream quantum kernel SVM (QKSVM) performance, exact empirical evaluation commonly relies on pairwise state-overlap information~\citep{NQE,TQE,TQE_physicalreview,IOP_NSGA2,GA_moretti,neuro_serach_iclr}. Therefore, exact evaluation over $B$ samples requires $O(B^2)$ pairwise quantum fidelity evaluations, and this cost causes computational bottleneck of the repeating optimization process. Subsampling and learned surrogate models can reduce this overhead, but rely on partial objective estimates, or on the performance of the surrogate model~\citep{sub_qkernel,neuro_serach_iclr}.

To address this bottleneck, we propose a prototype-guided framework that replaces dense sample-pair comparisons with sample-to-prototype relations. Prototype separation and sample-to-prototype compactness provide a measurable lower certificate for class-ensemble distinguishability and motivate an objective requiring $O(B)$ fidelity evaluations for $B$ samples with fixed number of prototypes. We further construct a prototype-based reproducing kernel Hilbert space (RKHS)~\citep{rkhs_paulsen} classifier as an explicit witness for the soft-margin SVM objective~\citep{svm_obj} and evaluate it as a low-cost proxy for model selection. The framework combines genetic algorithm (GA)~\citep{GA} architecture search with Neural Quantum Embedding (NQE)~\citep{NQE} refinement and extends to multiple prototypes for multimodal class structures. Across five benchmark datasets with up to 20 qubits, the method retains competitive QKSVM performance with lower fidelity-evaluations. And the advantage is observed on T-MNIST. Analyses of our theoretical lower bound for class-ensemble distinguishability, centered kernel-target alignment (CKTA)~\citep{CKTA}, and sample-wise SVM margins clarify both the observed gains and the limitations of mean-based objective, while the multi-prototype extension improves classification on a multimodal benchmark.

\section{Prototype-Guided Quantum Data Embedding}
\label{sec:methods}
\subsection{Quantum Data Embedding and Pairwise Optimization}
\label{subsec:ensemble_diff}
Consider a binary classification problem with a labeled dataset $\mathcal S=\{(\mathbf{x}_i,y_i)\}_{i=1}^{m}$, where $\mathbf{x}_i\in\mathcal X$ denotes a classical input and $y_i\in\{+,-\}$ is its corresponding class label. Let $I_y=\{i:y_i=y\}$ denote the index set of class $y$. We denote $m_y$ the number of samples belonging to class $y$, and $p_y=m_y/m$ be the corresponding empirical class prior. Given an $n$-qubit quantum data embedding circuit $U_{\Phi}(\mathbf{x})$, each input $\mathbf{x}_i$ is mapped to the pure quantum state $ |\psi_{\Phi}(\mathbf{x}_i)\rangle =  U_{\Phi}(\mathbf{x}_i)|0\rangle^{\otimes n}$, with associated density matrix $\rho_{\Phi}(\mathbf{x}_i) = |\psi_{\Phi}(\mathbf{x}_i)\rangle \langle\psi_{\Phi}(\mathbf{x}_i)|$.

After embedding, the samples in each class construct an empirical class-ensemble state, $\bar\rho_y = \frac{1}{m_y}\sum_{i:y_i=y}\rho_{\Phi}(\mathbf{x}_i).$ The trace norm of the prior-weighted difference,
\begin{equation}
\label{ensemble_diff}
    \Delta_{\Phi}
    =
    \left\|
    p_+\bar{\rho}_+
    -
    p_-\bar{\rho}_-
    \right\|_1 = 2D_{\mathrm{tr}}(p_+\bar{\rho}_+, p_-\bar{\rho}_-),
\end{equation}
quantifies the distinguishability between the two class ensembles~\citep{Bae_2015} where $D_\mathrm{tr}(\cdot,\cdot)$ denotes trace distance between two quantum states. $\Delta_\Phi$ has a direct operational meaning in binary quantum-state discrimination. In particular, it determines a lower bound on the empirical misclassification risk~\citep{qml_generalization_banchi,NQE}, and upper bounds the mean classification margin associated with margin-based generalization~\citep{quantum_margin}. Moreover, $\Delta_{\Phi}$ cannot be increased by subsequent quantum processing because the trace norm is contractive under completely positive trace-preserving maps~\citep{NQE}. This quantity is therefore an important embedding-level quantity that governs classification performance and generalization in QML.

Evaluating $\Delta_{\Phi}$ directly is expensive because it requires spectral information about the mixed-state difference operator. This motivates the use of a Hilbert--Schmidt (HS) relaxation which expresses the class-ensemble difference in terms of within-class purities and cross-class overlap~\citep{TQE,NQE}. For pure-state embeddings, these quantities are reduced to pairwise fidelity kernel $k_\Phi(\mathbf{x}_i,\mathbf{x}_j) = \operatorname{Tr}[\rho_\Phi(\mathbf{x}_i)\rho_\Phi(\mathbf{x}_j)]$. $\Delta_\Phi$ therefore admits the following lower bound.
\begin{equation}
    \Delta_{\Phi}
    \geq
    \frac{1}{m}
    \Bigg[
    \underbrace{
    \sum_{y\in\{+,-\}}
    \sum_{i,j\in I_y}
    k_{\Phi}(\mathbf{x}_i,\mathbf{x}_j)
    }_{\text{within-class purity terms}}
    -
    \underbrace{
    2
    \sum_{i \in I_+}
    \sum_{j \in I_-}
    k_{\Phi}(\mathbf{x}_i,\mathbf{x}_j)
    }_{\text{cross-class overlap term}}
    \Bigg]^{1/2}.
    \label{eq:pairwise_ensemble_lower_bound}
\end{equation}

The derivation is in Appendix~\ref{app:hs_relaxation}. This inequality identifies a class-discriminative geometry where same-class samples cluster while different-class samples remain well separated. The perspective provides a common geometric interpretation of existing data-adaptive quantum embedding methods. For multiclass datasets, we apply the analysis to each class pair, with $m$ and the empirical priors defined within that pair, and report the mean pairwise class-ensemble distinguishability.

Even when the connection between this geometry and class-ensemble distinguishability is not made explicit, many data-adaptive embedding methods pursue closely related objectives. \citet{TQE} and \citet{NQE} explicitly connect this class-ensemble distinguishability to embedding optimization, by optimizing either its HS relaxation or a related fidelity-based objective. Other data-adaptive embedding methods do not make this connection explicit, but employ conceptually similar criteria based on quantum kernel--target alignment~\citep{TQE_physicalreview,neuro_serach_iclr} or downstream QKSVM performance~\citep{IOP_NSGA2,GA_moretti}. 

For $B$ samples, exact evaluation requires $\binom{B}{2} = \frac{B(B-1)}{2} \in O(B^2)$ distinct pairwise kernel evaluations.
This cost is incurred repeatedly during embedding optimization, creating a computational bottleneck. Pair subsampling can reduce the number of evaluated kernel entries, but optimizes only a partial or stochastic estimate of the full pairwise geometry~\citep{NQE,sub_qkernel}. Learned surrogate models can further reduce direct quantum-kernel evaluations, but their effectiveness depends on the accuracy of the surrogate predictor~\citep{neuro_serach_iclr}. These limitations motivate an efficiently computable lower certificate for class-ensemble distinguishability.

\begin{figure}[t]
    \centering
    \begin{minipage}{0.84\linewidth}
        \centering
        \includegraphics[width=\linewidth]{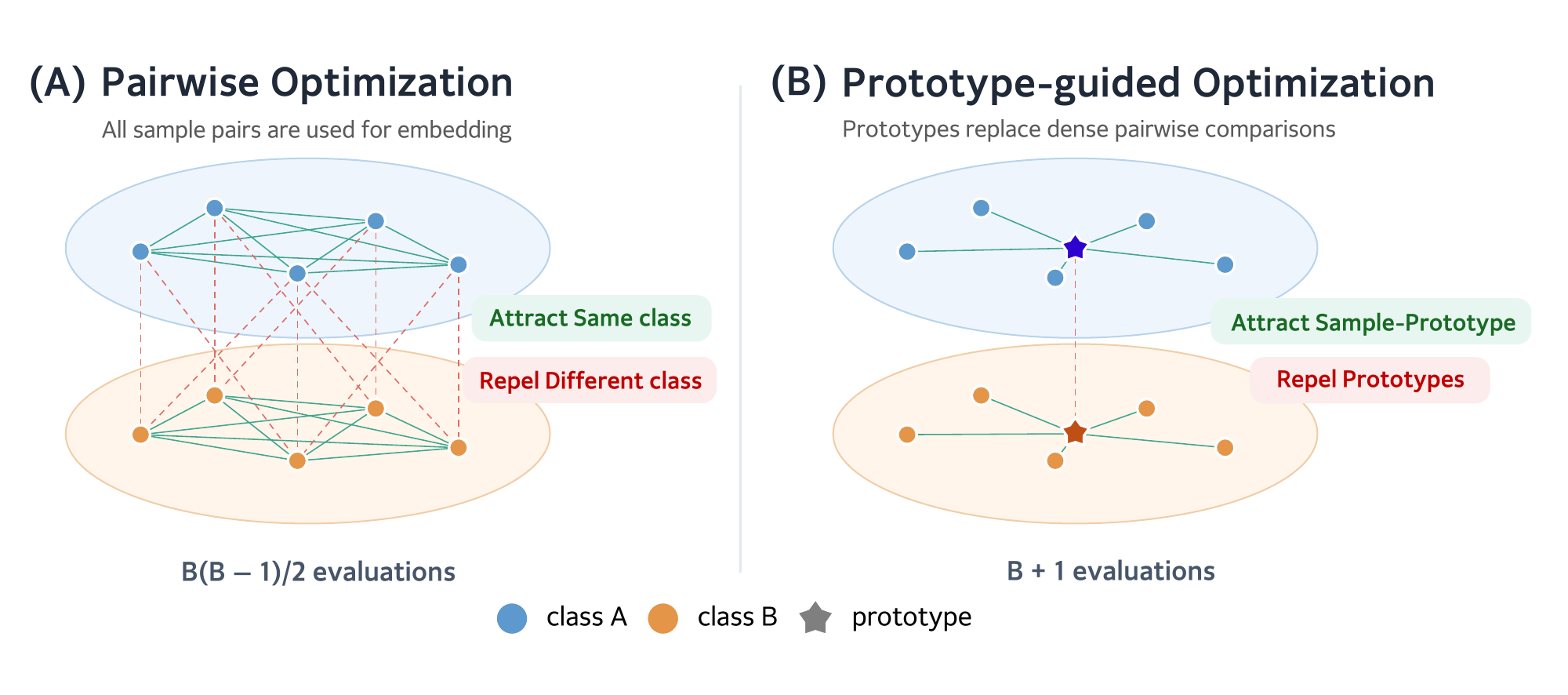}
        \vspace{-10mm}
        \caption{
        Comparison of pairwise and single-prototype-guided quantum embedding optimization.
        \textbf{(a)} Pairwise optimization promotes high fidelity for same-class sample pairs and low fidelity for cross-class pairs, requiring $B(B-1)/2$ evaluations.
        \textbf{(b)} The proposed method promotes high sample-to-class-prototype fidelity and low cross-class prototype fidelity, reducing the cost to $B+1$ evaluations.
        }
        \label{fig:pairwise_prototype}
    \end{minipage}
\end{figure}

\subsection{Prototype-Based Ensemble Certificate}
\label{subsec:prototype_ceritficate}

To construct such a certificate, we consider one prototype per class, $\{\mathbf{c}_+,\mathbf{c}_-\} \in \mathcal{X}$ with the application-dependent construction rule. The prototype quantum state of class $y$ is defined as $\rho_y^c=\rho_\Phi(\mathbf{c}_y)$. The prototype fidelity therefore becomes $F_c=k_\Phi(\mathbf{c}_+,\mathbf{c}_-)$ and their trace distance is $D_c=D_{\mathrm{tr}}(\rho_+^c,\rho_-^c)=\sqrt{1-F_c}$ since we consider pure state embedding.

We define the mean within-class trace-distance radius to quantify within-class compactness around the prototypes after embedding. It measures the average trace distance between the embedded samples and their corresponding prototypes:
\begin{equation}
\bar r_{y}^{\mathrm{tr}} =
\frac{1}{m_y}\sum_{i\in I_y}D_{\mathrm{tr}}\left(\rho_{\Phi}(\mathbf{x}_i),\rho_{y}^{c}\right)=
\frac{1}{m_y}\sum_{i\in I_y}\sqrt{1-k_{\Phi}(\mathbf{x}_i,\mathbf{c}_y)}.
\label{eq:radius}
\end{equation}

The prototype fidelity $F_c$ and the class-wise radii $\bar r_y^{\mathrm{tr}}$ respectively quantify inter-class prototype separation and within-class compactness. Together, they yield the following lower certificate for class-ensemble distinguishability $\Delta_{\Phi}$.

\begin{proposition}[Prototype-based Ensemble Certificate]
\label{prop1}
For arbitrary class prototypes $\mathbf{c}_{+}$ and $\mathbf{c}_{-}$, class-ensemble distinguishability $\Delta_{\Phi}$ satisfies
\begin{equation}
\left\lVert p_{+}\bar{\rho}_{+} - p_{-}\bar{\rho}_{-} \right\rVert_1
\geq
\sqrt{1-4p_{+}p_{-}F_c} - 2 \left( p_{+}\bar r_{+}^{\mathrm{tr}} + p_{-}\bar r_{-}^{\mathrm{tr}} \right).
\label{eq:prototype_certificate}
\end{equation}
\end{proposition}

Proposition~\ref{prop1} shows that the distinguishability of the full class ensembles can be lower-bounded using quantities determined only by the prototypes and their associated samples. Specifically, the first term captures the separation between the two prototype states, while the second penalizes the dispersion of each class around its prototype. Thus, well-separated prototypes together with compact class-wise embeddings yield a larger lower certificate for $\Delta_{\Phi}$. The result follows Appendix~\ref{appendix:prop1_proof}. For later use, we denote the right-hand side of Eq.~\ref{eq:prototype_certificate} by the prior-weighted prototype certificate,
\begin{equation}
\mathcal{C}_{\mathrm{prior}}
=
\sqrt{1-4p_{+}p_{-}F_c}
-2\left(
p_{+}\bar r_{+}^{\mathrm{tr}}
+
p_{-}\bar r_{-}^{\mathrm{tr}}
\right).
\label{eq:c_prior}
\end{equation}

A single prototype may not adequately represent a class with multimodal structure. We thus associate each class $y\in\{+,-\}$ with a set of $K_y$ prototypes, $\mathcal{P}_y = \{\mathbf{c}_{y,1},\ldots,\mathbf{c}_{y,K_y}\}.$ Each sample \(x_i\) is assigned to a prototype of its class according to a predefined rule, with assignment \(a(i)\in\{1,\ldots,K_{y_i}\}\) fixed before optimization.

This assignment partitions each class into prototype-specific groups $I_{y,a}=\{i\in I_y:a(i)=a\}$, whose empirical proportions $\pi_{y,a}=|I_{y,a}|/m_y$ determine the contribution of each prototype state
$\rho_{y,a}^{c}=\rho_{\Phi}(\mathbf{c}_{y,a})$. The resulting prototype mixture for class $y$ is
\begin{equation}
\widetilde{\rho}_{y}^{c}
=
\sum_{a=1}^{K_y}
\pi_{y,a}\rho_{y,a}^{c}
=
\frac{1}{m_y}
\sum_{i\in I_y}
\rho_{y,a(i)}^{c}.
\label{eq:prototype_mixture}
\end{equation}

Extending Eq.~\ref{eq:radius} to the multi-prototype setting, we define the mean trace distance between the samples of class $y$ and their respective prototypes as
\begin{equation}
\bar r_{y,K_y}^{\mathrm{tr}}
=
\frac{1}{m_y}
\sum_{i\in I_y}
D_{\mathrm{tr}}
\left(
\rho_{\Phi}(\mathbf{x}_i),
\rho_{y,a(i)}^{c}
\right).
\label{eq:multi_radius}
\end{equation}
Using these definitions, we extend Proposition~\ref{prop1} to the multi-prototype setting.

\begin{proposition}[Multi Prototype-based Ensemble Certificate]
\label{prop:multi}
For the class-wise prototype sets $\mathcal{P}_y=\{\mathbf{c}_{y,1},\ldots,\mathbf{c}_{y,K_y}\}$, the class-ensemble distinguishability $\Delta_{\Phi}$ satisfies
\begin{align}
\left\|
p_{+}\bar\rho_{+}
-
p_{-}\bar\rho_{-}
\right\|_{1}
&\geq
\|p_+\tilde\rho_+^c-p_-\tilde\rho_-^c\|_1
-
2\sum_{y\in\{+,-\}}p_y\bar r_{y,K_y}^{\mathrm{tr}}.
\label{eq:multi_certificate}
\end{align}
\end{proposition}
The proof is provided in Appendix~\ref{appendix:prop2_proof}.

The resulting bound provides a lower certificate for $\Delta_{\Phi}$ based on the geometry of the prototype mixtures and the average distances between samples and their assigned prototypes. The prototype-mixture separation term can be viewed as a prototype-level counterpart of the pairwise geometry in Section~\ref{subsec:ensemble_diff}. Its fidelity-based expansion favors large separation between prototypes from different classes and coherence among prototypes within the same class. The radius term further encourages each sample to remain close to its assigned prototype. For multiclass datasets, Propositions~\ref{prop1} and~\ref{prop:multi} are applied to each class pair, and averaging the resulting bounds yields a lower certificate for the mean pairwise class-ensemble distinguishability. These geometric conditions motivate the practical prototype-guided objective introduced next.

\subsection{Prototype-Guided Optimization}
\label{subsec:prototype_optimization}

Motivated by Proposition~\ref{prop:multi}, we construct a practical fidelity-based objective with three components: (1) sample-to-prototype compactness $\mathcal{L}_\text{sample}$, (2) within-class prototype coherence $\mathcal{L}_\text{in}$, and (3) cross-class prototype separation $\mathcal{L}_\text{out}$. Since directly evaluating the trace-distance quantities in the certificate is generally costly on quantum hardware, we follow prior quantum embedding approaches~\citep{TQE,NQE} and use fidelity as practical surrogates for the corresponding geometry.

Let $\mathcal{Y}$ be a general class set and $y,z\in\mathcal{Y}$ be index the classes. With these notations, the three components identified before are defined as
\begin{align}
\mathcal{L}_{\mathrm{sample}}
&=
\sum_{y\in\mathcal{Y}}
\sum_{i\in I_y}
\left[
1-k_{\Phi}(\mathbf{x}_i,\mathbf{c}_{y,a(i)})
\right],
\label{eq:L_sample}
\\
\mathcal{L}_{\mathrm{in}}
&=
\sum_{y\in\mathcal{Y}}
\sum_{a,a'=1}^{K_y}
\left[
1-k_{\Phi}(\mathbf{c}_{y,a},\mathbf{c}_{y,a'})
\right],
\label{eq:L_in}
\\
\mathcal{L}_{\mathrm{out}}
&=
\sum_{y<z}
\sum_{a=1}^{K_y}
\sum_{b=1}^{K_z}
k_{\Phi}(\mathbf{c}_{y,a},\mathbf{c}_{z,b}).
\label{eq:L_out}
\end{align}
Accordingly, decreasing $\mathcal{L}_{\mathrm{in}}$ and $\mathcal{L}_{\mathrm{out}}$ promotes the prototype-level geometry underlying the certificate, while $\mathcal{L}_{\mathrm{sample}}$ provides a squared-distance surrogate for the sample-to-prototype radius penalty.

We collect the three components into the prototype-guided objective
\begin{equation}
\mathbf{\mathcal{L}}_{\mathrm{PG}}(\vartheta)
=
\left(
\mathcal{L}_{\mathrm{sample}}(\vartheta),
\mathcal{L}_{\mathrm{in}}(\vartheta),
\mathcal{L}_{\mathrm{out}}(\vartheta)
\right),
\label{eq:prototype_objective}
\end{equation}
where $\vartheta$ denotes the optimization variables that determine the quantum embedding. The formulation is agnostic to any optimization algorithm used to balance the three components. In our experiments, they are combined using a weighted sum.

Prototype assignments are determined once in the classical representation space before embedding optimization and kept fixed across candidate embeddings. After embedding, each sample is compared only with its assigned prototype so that $\mathcal{L}_{\mathrm{sample}}$ requires $B$ fidelity evaluations. For $\mathcal{L}_{\mathrm{in}}$, the within-class prototype terms require $\sum_y\binom{K_y}{2}$ fidelity evaluations, while $\mathcal{L}_{\mathrm{out}}$ requires all $\sum_{y<z}K_yK_z$ cross-class prototype pairs. Hence, the exact number of distinct fidelity evaluations is
\begin{equation}
N_{\mathrm{fid}}
=
B
+
\sum_y\binom{K_y}{2}
+
\sum_{y<z}K_yK_z = B+\binom{K_\mathrm{tot}}{2}, \quad K_{\mathrm{tot}}=\sum_{y\in \mathcal Y} K_y.
\label{eq:prototype_eval_cost}
\end{equation}
Thus, the proposed objective requires $O(B+K_{\mathrm{tot}}^2)$ fidelity evaluations. For $K_{\mathrm{tot}}$ fixed independently of $B$, this reduces to $O(B)$, compared with $O(B^2)$ for full pairwise optimization.

\subsection{Prototype RKHS Witness for Quantum Kernel Classification}
\label{subsec:rkhs_to_svm}

Beyond improving the efficiency of embedding optimization, the prototype construction also provides a direct connection between the optimized embedding geometry and downstream classification. We make this connection explicit through the RKHS~\citep{rkhs_paulsen} induced by the fidelity kernel, using the soft-margin SVM as the downstream classifier. Although the analysis below is specific to the SVM, the underlying prototype geometry is not restricted to kernel classifiers: for a fixed embedding, both kernel classifiers and variational quantum classifiers act as linear decision functions on the embedded quantum-state feature space~\citep{q_featuremap}.

We now return to the single-prototype setting and define the prototype-based RKHS classifier
\begin{equation}
    h_{\text{proto}}(\mathbf{x}) = k_{\Phi}(\mathbf{x}, \mathbf{c}_+) - k_{\Phi}(\mathbf{x}, \mathbf{c}_-).
\end{equation}
By the reproducing property, $h_{\text{proto}} \in \mathcal{H}_{\Phi}$, and its squared RKHS norm is $||h_{\text{proto}}||^2_{\mathcal{H}_{\Phi}} = 2(1-F_c) = 2D_c^2$. For the embedding $\Phi$, we consider the normalized soft-margin SVM objective~\citep{svm_obj} with resepectly identifying the class labels $+$ and $-$ with $+1$ and $-1$,
\begin{equation}
\mathcal{J}_{\mathrm{SVM}}^{\star}(\Phi)
=
\inf_{\substack{
f\in\mathcal{H}_{\Phi}\\
b\in\mathbb{R}
}}
\left\{
\frac{\lambda_{\mathrm{SVM}}}{2}
\|f\|_{\mathcal{H}_{\Phi}}^{2}
+
\frac{1}{m}
\sum_{i=1}^{m}
\left[
1-y_i\left(f(\mathbf{x}_i)+b\right)
\right]_{+}
\right\}.
\label{eq:svm_objective}
\end{equation}
where $\lambda_{\mathrm{SVM}}>0$ denotes the regularization parameter. Using $h_{\mathrm{proto}}$ as an explicit RKHS witness, we establish the following upper bounds on its optimal value.

\begin{proposition}[Prototype RKHS Witness for the SVM Objective]
\label{prop:svm_witness}
For a fixed quantum embedding $\Phi$, the optimal value of the normalized soft-margin SVM objective satisfies
\begin{align}
\mathcal{J}_{\mathrm{SVM}}^{\star}(\Phi)
&\leq
\inf_{\tau\geq 0}
\Bigg\{
\frac{\lambda_{\mathrm{SVM}}}{2}
\tau^{2}
\|h_{\mathrm{proto}}\|_{\mathcal{H}_{\Phi}}^{2}
+
\frac{1}{m}
\sum_{i=1}^{m}
\left[
1-\tau y_i h_{\mathrm{proto}}(\mathbf{x}_i)
\right]_{+}
\Bigg\}.
\label{eq:svm_witness_bound}
\end{align}
\end{proposition}

The inequality provides a computable upper-bound surrogate for the optimal SVM objective by the prototype-based RKHS classifier $h_\text{proto}$. Therefore, an embedding that yields a well-performing prototype classifier also provides evidence that a good SVM classifier can be constructed. Computing this bound requires evaluating each sample against both class prototypes, resulting in $O(2B)$ fidelity evaluations. Therefore, it provides a cheaper indicator of the downstream linear classification quality. The extension to multiple prototypes is provided in Appendix~\ref{appendix:multi_proto_witness}.

\section{Results}
\label{sec:results}

\begin{table}[t]
    \centering
    \begin{minipage}{0.95\textwidth}
        \centering
        \caption{Held-out QKSVM performance (MCC/AUC, 10-seed means). The ZZfeaturemap applies $ZZ$ rotations only between nearest-neighbor qubits on a rectangular lattice. Full pairs provides the $O(B^2)$ reference, while Pairwise-$B$ is the budget-matched $O(B)$ pairwise baseline. The GA$\rightarrow$NQE panel compares Pairwise-B and Prototype after NQE refinement of the GA-selected embeddings.}
        \label{tab:main_results}
    \end{minipage}
    \scriptsize
    \renewcommand{\arraystretch}{1.1}
    \begin{tabular}{llccccc}
        \toprule
        \multirow{2}{*}{Stage} & \multirow{2}{*}{Method}
        & IRIS 
        & VERTEBRAL
        & T-CELL
        & GLASS 
        & T-MNIST \\
        & & \multicolumn{5}{c}{Test MCC / Test AUC} \\

        \midrule
        Fixed
        & ZZFeatureMap
        & .6470 / .8842
        & .1981 / .6749
        & .4542 / .7686
        & .3553 / .8025
        & .0703 / .5707 \\

        \midrule

        \multirow{3}{*}{GA}
        & Full pairs    $ [O(B^2)]$
        & \textbf{.9554 / .9994}
        & .6152 / \textbf{.9234}
        & \textbf{.5548} / .8034
        & \textbf{.6057} / .8723
        & -- \\

        & Pairwise-$B$  $[O(B)]$
        & .9200 / .9966
        & \textbf{.6403} / .9125
        & .5427 / .8141
        & .5907 / .8603
        & .6357 / .8883 \\

        & Prototype     $[O(B)]$
        & .9387 / .9964
        & .5939 / .8851
        & .5502 / \textbf{.8283}
        & .5979 / \textbf{.8837}
        & \textbf{.7606 / .9437} \\

        \midrule

        \multirow{2}{*}{GA$\rightarrow$NQE}
        & Pairwise-$B$ $[O(B)]$
        & .9419 / \textbf{.9976}
        & \textbf{.6471 / .9093}
        & \textbf{.5624} / .8007
        & .5826 / .8729
        & .6423 / .8962 \\

        & \textbf{Prototype $[O(B)]$}
        & \textbf{.9574} / .9975
        & .6460 / .8978
        & .5543 / \textbf{.8325}
        & \textbf{.5919 / .8765}
        & \textbf{.7787 / .9516} \\

        \bottomrule
    \end{tabular}
\end{table}

In this section, we evaluate our prototype-guided quantum embedding method across five benchmark datasets---IRIS \citep{iris}, T-CELL \citep{tcell}, VERTEBRAL \citep{vert}, GLASS \citep{glass}, tailored-MNIST~\citep{mnist}---spanning different dimensionalities and class structures. Dataset details are provided in Appendix~\ref{appendix:data}. The quantum embeddings use 4, 6, 8, 9, and 20 qubits for IRIS, VERTEBRAL, T-CELL, GLASS, and T-MNIST, respectively. Optimization is performed using only the training set, while model selection uses the validation set. The test set is reserved exclusively for final evaluation. For each learned embedding, we construct the quantum fidelity kernel and evaluate QKSVM~\citep{SVM}. The Matthews correlation coefficient (MCC) is used as the primary evaluation metric because it is applicable to multiclass settings and robust to class imbalance. The area under the receiver operating characteristic curve (ROC--AUC) is reported as a secondary metric for the same QKSVM model selected by its MCC. Unless otherwise stated, all experiments are conducted using noiseless statevector simulation in \textit{Qiskit}~\citep{qiskit}.


\subsection{Single Prototype Optimization Results}
\label{subsec:single_prototype_results}

\begin{figure}[t]
    \centering
    \includegraphics[width=1\linewidth]{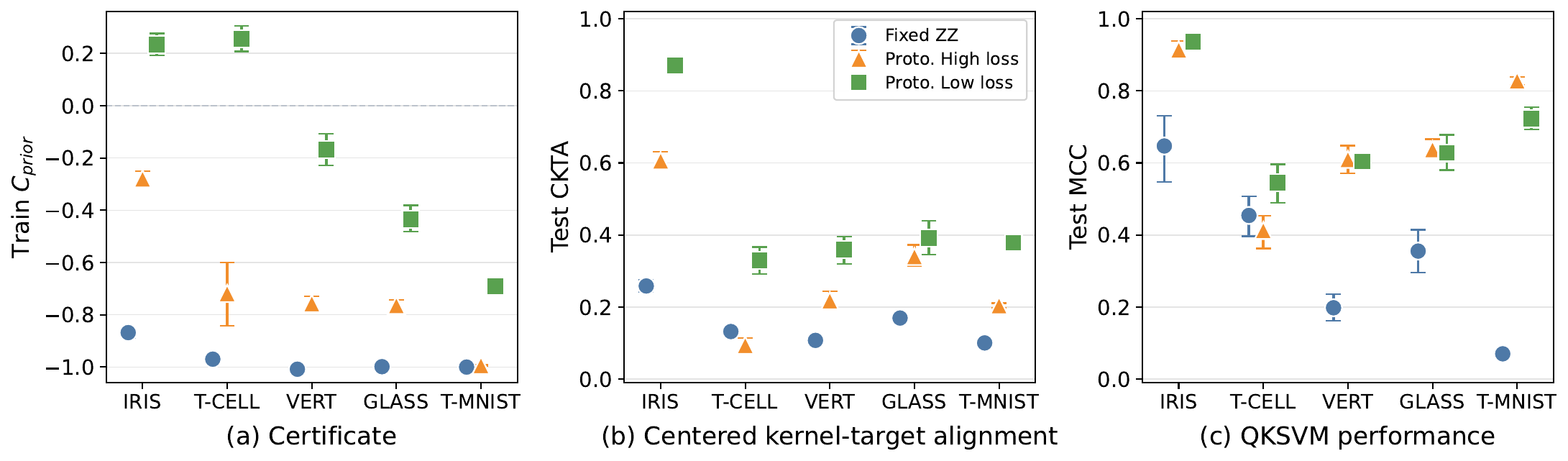}
    \begin{minipage}{1\linewidth}
        \centering
        \caption{Prototype geometry and downstream behavior. VERT means VERTEBRAL.
        For each seed, ten unique GA circuits with the lowest and highest
        training prototype losses form the low- and high-loss groups,
        respectively.
        \textbf{(a)} Training $C_{\mathrm{prior}}$, where the dashed line
        marks the threshold for a non-vacuous positive certificate.
        \textbf{(b)} Test pairwise-macro CKTA.
        \textbf{(c)} Test QKSVM MCC.
        Markers and error bars denote the mean and 95\% seed-bootstrap CI.
        }
        \label{fig:certificate}
    \end{minipage}
\end{figure}

To evaluate the efficiency and predictive performance of the proposed objective, we first use a GA~\citep{GA} to search for a data-adaptive circuit architecture. The circuit architecture is optimized progressively across increasing depths, with the best circuit found at each depth used to initialize the search at the next depth. This layer-wise strategy reuses previously optimized circuit structure rather than restarting the architecture search independently at each depth. After the GA, We further apply NQE~\citep{NQE} to the selected circuit, keeping its architecture fixed while continuously refining the data embedding. A residual neural network~\citep{residual} predicts additive corrections to the original feature values, with its parameters trained to minimize the corresponding embedding objective through gradient-based optimization. Together, these stages yield an end-to-end optimization pipeline combining discrete architecture search with continuous embedding refinement. Detailed settings are provided in Appendix~\ref{appendix:experiments_details}.

We compare the proposed method with three quantum embedding baselines in Table~\ref{tab:main_results}. The fixed ZZ feature map \citep{zzfeaturemap} serves as a standard data-independent embedding, with its circuit depth selected according to validation MCC. Full pairwise optimization, described in Section~\ref{subsec:ensemble_diff}, evaluates the fidelity objective over all training-sample pairs and serves as the $O(B^2)$ reference. We also include a budget-matched pairwise baseline (Pairwise-$B$) that randomly samples a subset of training pairs so that each objective evaluation uses the same number of fidelity evaluations as the proposed method. For a controlled comparison, all methods use the same genetic-algorithm configuration and quantum-circuit search space, differing only in their optimization objectives and fidelity-evaluation budgets. We omit full-pairwise NQE refinement due to the computational cost of differentiating through all $O(B^2)$ fidelities at each update in our experimental setting. The GA$\rightarrow$NQE comparison therefore includes only the proposed and budget-matched objectives.

Table~\ref{tab:main_results} summarizes the held-out QKSVM performance after GA search and subsequent NQE refinement. At the GA stage, the prototype-guided objective remains competitive with full-pairwise optimization while requiring only $O(B)$ fidelity evaluations, and consistently improves upon the fixed ZZ feature map. Under the same $O(B)$ fidelity budget, its advantage is particularly pronounced on the 20-qubit T-MNIST task, improving MCC from $.6357$ to $.7606$ over Pairwise-$B$. After NQE refinement, the prototype objective achieves higher MCC than Pairwise-$B$ on IRIS, GLASS, and T-MNIST, with the largest gap again observed on T-MNIST ($.7787$ vs. $.6423$). 

Further analysis of the final T-MNIST embeddings shows that the prototype-guided method yields greater pairwise class-ensemble
distinguishability and larger RKHS-normalized QKSVM margins for most test samples than Pairwise-$B$, providing geometric evidence
consistent with its performance advantage (Appendix~\ref{app:t_mnist_good}). Overall, these results demonstrate that prototype-guided optimization can substantially reduce fidelity-evaluation cost while retaining competitive QKSVM performance, although its relative benefit remains dataset dependent.

\subsection{Empirical Behavior of the Certificate and RKHS Witness}
\label{subsec:certificate_and_rkhs}

To examine whether the geometry promoted by the prototype objective transfers beyond its training loss, we compare ZZ feature map with ten circuits having the lowest and highest training prototype losses during GA optimization. The groups are constructed using only the train objective. We evaluate $C_{\mathrm{prior}}$ on the train set, CKTA on the test set, and the test MCC of the corresponding QKSVM.

As shown in Figure~\ref{fig:certificate}(a)--(b), the low-loss circuits consistently yield the largest $C_{\mathrm{prior}}$ and CKTA across all datasets. They also achieve higher mean MCC than the ZZ feature map, with the clearest separation on IRIS and T-CELL, where $C_{\mathrm{prior}}$ becomes positive. On the remaining datasets, the certificate remains negative and therefore provides no improvement over the trivial zero lower bound.

Higher CKTA does not, however, translate monotonically into MCC, particularly on T-MNIST. While CKTA summarizes global Gram-matrix alignment, QKSVM performance also depends on margins near the decision boundary. Relative to the high-loss circuits, the low-loss circuits improve RKHS-normalized geometric test margins for most samples but worsen them for a majority of samples in the hardest difficulty decile (Appendix~\ref{app:t_mnist_bad}). This pattern suggests that the mean-based objective can favor easier samples at the expense of difficult boundary samples, illustrating why improved average embedding geometry need not yield better QKSVM performance.

We next test whether the prototype construction can be used for validation candidate selection. The 20 prototype-optimized circuits in the bank of GA is reranked using the validation prototype loss (P0), the exact prototype RKHS classifier in Proposition~\ref{prop:svm_witness} (P1) and a full validation QKSVM (P2).

\begin{table*}[t]
\centering
\begin{minipage}{0.84\textwidth}
\centering
\caption{Held-out MCC/AUC after validation reranking of the same dataset-specific bank of $20$ prototype-optimized circuits (10-seed means). Here, $N$ and $V$ denote the numbers of
training and validation samples, respectively, and $C$ denotes the number of classes, with one prototype per class. P0 selects using validation prototype loss, P1 using the exact prototype RKHS classifier in Proposition~\ref{prop:svm_witness}, and P2 using a full validation QKSVM. $\Delta$MCC denotes P1--P0, and the leading fidelity cost of each selector is shown in the parentheses, the exact cost is in Appendix~\ref{app:selection}.}
\label{tab:table2}
\end{minipage}
\setlength{\tabcolsep}{2.4pt}
\setlength{\tabcolsep}{4.0pt}
\begin{tabular}{l cc c c}
\toprule
&
\multicolumn{2}{c}{Linear-cost selectors}
&
&
\multicolumn{1}{c}{Full-QKSVM reference}
\\
\cmidrule(lr){2-3}
\cmidrule(lr){5-5}
Dataset
&
\shortstack{P0 $(V)$}
&
\shortstack{P1 $(CV)$}
&
$\Delta$ MCC
&
\shortstack{P2 $(N^2)$}
\\
\midrule

IRIS
& $.9387/.9964$
& $.9387/.9967$
& $+.0000$
& $.9387/.9976$
\\

VERTEBRAL
& $.5939/.8851$
& $.6189/.8920$
& $\mathbf{+.0250}$
& $.6201/.8946$
\\

T-CELL
& $.5502/.8283$
& $.5502/.8283$
& $+.0000$
& $.5502/.8283$
\\

GLASS
& $.5979/.8837$
& $.6073/.8800$
& $\mathbf{+.0094}$
& $.6001/.8865$
\\      

T-MNIST
& $.7606/.9437$
& $.7557/.9427$
& $-0.0049$
& $.7657/.9446$
\\

\bottomrule
\end{tabular}
\end{table*}

Table~\ref{tab:table2} shows that P1 improves or preserves the mean test MCC relative to P0 on four of five datasets, and matches or exceeds P2 on three. These results do not identify the prototype classifier with the trained QKSVM. They rather show that it provides a classifier-aware RKHS witness that can serve as a useful proxy. When the additional validation cost is affordable, prototype classifier offers a more informative selector than the prototype loss without requiring the quadratic train kernel.

\subsection{Multi Prototype Opimization Results}
\label{sec:multi_result}

A single prototype may be insufficient when a class contains multiple distinct modes. We therefore evaluate the multi-prototype extension of Proposition~\ref{prop:multi} on GLASS. To alleviate its class sparsity, the original glass types are grouped into three broader classes: float-processed window glass, non-float-processed window glass, and non-window glass. These three classes contain two, one, and three original sublabels, respectively. We therefore use $K=(2,1,3)$, where each entry denotes the number of prototypes for the corresponding class in the order listed above. One prototype is constructed per subclass, and each sample remains assigned to its subclass prototype throughout optimization. Thus, the experiment directly tests whether representing distinct within-class modes by separate prototypes improves the geometry targeted by the multi-prototype objective.

As a result, the multi-prototype embedding increases the test MCC from $0.592$ to $0.647$, while the AUC remains similar ($0.876$ to $0.884$). Figure~\ref{fig:multi_k}(a) compares the resulting QKSVM with the single-prototype model. The out-of-fold (OOF) margin distribution also shifts toward larger values overall, although the multi-prototype model retains a somewhat longer negative tail. This indicates that the improvement is not explained by uniformly increasing every sample margin. Rather, the multi-prototype representation improves the margin geometry for a substantial portion of the samples while a smaller set of difficult boundary samples remains.

To examine the mechanism more directly, Figure~\ref{fig:multi_k}(b) relates the sample-wise trace distance reduction to the assigned prototype, $r_{K=1}-r_{\mathrm{multi}}$, to the corresponding OOF margin gain, $m_{\mathrm{multi}}-m_{K=1}$. The two quantities exhibit a positive association ($\rho_s=0.33$) that implies samples whose local prototype radius is reduced more tend to obtain larger classification margins. Although the association is moderate rather than deterministic, its direction is consistent across the quantile bins and supports the mechanism underlying Proposition~\ref{prop:multi}. These results provide empirical evidence that multiple prototypes can improve local class coverage and translate the reduction in sample-to-prototype radius into a more favorable QKSVM decision geometry.

\begin{figure}[t]
    \centering
    \includegraphics[width=.85\linewidth]{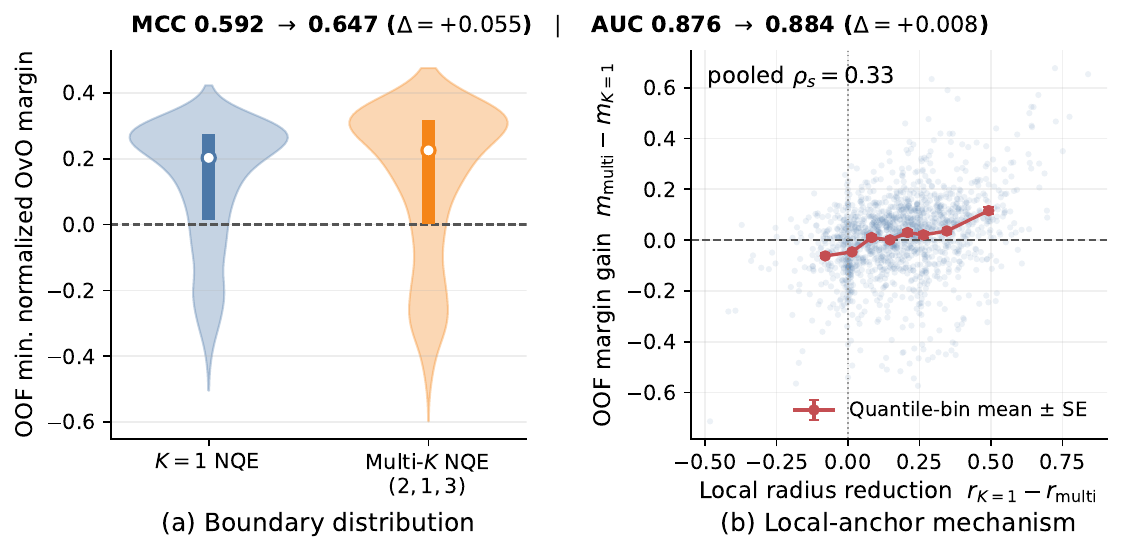}
    \begin{minipage}{0.85\linewidth}
        \centering
        \caption{Multi-prototype optimization on GLASS.
        \textbf{(a)} Vilolin plots of OOF minimum normalized one-vs-one margin distributions for the single-prototype ($K=1$) and multi-prototype ($K=(2,1,3)$) model, where the entries of $K$ denote the class-specific numbers of prototypes. The white circle and the bar indicate the median and interquartile range (IQR) of each distribution, respectively.
        \textbf{(b)} OOF margin gain versus the reduction in sample-to-assigned-prototype trace-distance radius. Quantile-bin means $\pm$ standard errors are shown. $\rho_s$ means the pooled Spearman correlation.
        }
        \label{fig:multi_k}
    \end{minipage}
\end{figure}

\section{Summary and Discussion}
\label{sec:summary}

We introduced a prototype-guided framework for efficient data-adaptive quantum data embedding optimization. By replacing dense sample-pair comparisons with prototype separation and sample-to-prototype compactness, the proposed method provides a measurable lower bound for class-ensemble distinguishability and an explicit RKHS witness for the soft-margin SVM objective. This reduces the fidelity-evaluation cost from $\mathcal{O}(B^2)$ to $\mathcal{O}(B)$ for a constant number of prototypes. 

Across the benchmark tasks, the resulting embeddings remained competitive with full-pairwise optimization while consistently improving upon the fixed feature map, showing that effective data-adaptive quantum embeddings do not necessarily require exhaustive pairwise evaluations.

The empirical analyses clarify the scope of this compression. Lower prototype loss consistently produced more favorable training certificates and held-out kernel alignment, but these improvements did not translate monotonically into QKSVM performance. This reflects the fact that $C_{\mathrm{prior}}$ is a lower bound for embedding-level class distinguishability rather than a direct characterization of the final SVM decision boundary. The prototype RKHS classifier provides a closer bridge to downstream classification: using it for validation improved or preserved MCC-based circuit selection on four of five datasets while retaining linear evaluation cost. We therefore interpret the certificate as a geometric diagnostic and the RKHS construction as a computationally inexpensive witness of classifier quality, rather than as exact predictors of the trained QKSVM.

The multi-prototype experiment further shows that multiple representatives can better capture multimodal class structure, improving GLASS MCC while reducing local sample-to-prototype distances. However, the gains remain dataset-dependent, and difficult boundary points are not fully explained by prototype geometry alone. More generally, the proposed fidelity objective is a surrogate for the trace-distance certificate, and its effectiveness depends on prototype quality and class structure. 
The present multi-prototype study also relies on known subclass information, leaving automatic prototype construction, adaptive assignment, and evaluation under finite-shot noisy hardware as important directions for future work.

\subsubsection*{AI Use Disclosure}
In this work, we used generative AI tools to create or edit software code and make up draft parts of a research paper. We also used the AI tools to summarize or analyze existing literature and edit a research paper to improve readability.

\subsubsection*{Ethics Statement}

This work focuses on the theoretical and methodological development of quantum data embeddings and their evaluation on existing benchmark datasets. We did not collect new personal data or conduct new experiments involving human participants. Dataset sources and preprocessing procedures are described in Appendix~\ref{appendix:data}. The biomedical benchmarks are used to evaluate classification performance, and the results do not establish clinical validity.

\subsubsection*{Reproducibility Statement}

To support reproducibility, we provide derivations and proofs of the theoretical results in Appendices~\ref{app:proof} and~\ref{appendix:multi_proto_witness}. Appendix~\ref{appendix:data} describes the datasets, label construction, preprocessing, and data split. Appendix~\ref{appendix:experiments_details} details the genetic algorithm search and neural quantum embedding refinement. Appendix~\ref{app:selection} specifies the validation-based circuit selection rules, including the prototype RKHS classifier, and their fidelity-evaluation costs. Appendix~\ref{app:tmnist} describes the additional margin analyses. All quantum simulations use noiseless statevector simulation, as specified in the experimental setup. 

\subsubsection*{Acknowledgments}
This work was supported by the Institute of Information \& Communications Technology Planning \& Evaluation (IITP) ITRC (Information Technology Research Center) funded by the the Ministry of Science and ICT (MSIT), Korea (IITP-2026-RS-2026-25519864), by the IITP grant (No. 2019-0-00003, Research and Development of Core Technologies for Programming, Running, Implementing and Validating of Fault-Tolerant Quantum Computing System), by the National Research Foundation of Korea (RS-2025-02309510), by Korean ARPA-H Project through the Korea Health Industry Development Institute (KHIDI), funded by the Ministry of Health \& Welfare, Korea (RS-2025-25456722), and National Institute of Health (NIH) research project (2026-ER0901-00). We thank Boseong Kim for stimulating discussions.


\appendix

\section*{Appendix}

\section{Proofs}
\label{app:proof}
\subsection{Derivation of the Hilbert--Schmidt Relaxation (Equation~\ref{eq:pairwise_ensemble_lower_bound})}
\label{app:hs_relaxation}

For a class ensemble distinguishability 
\begin{equation*}
\label{ensemble_diff}
    \Delta_{\Phi}
    =
    \left\|
    p_+\bar{\rho}_+
    -
    p_-\bar{\rho}_-
    \right\|_1 = 2D_{\mathrm{tr}}(p_+\bar{\rho}_+, p_-\bar{\rho}_-),
\end{equation*}
it admits the following lower bound
\begin{equation}
    \Delta_{\Phi}
    \geq
    \frac{1}{m}
    \Bigg[
    \underbrace{
    \sum_{y\in\{+,-\}}
    \sum_{i,j\in I_y}
    k_{\Phi}(\mathbf{x}_i,\mathbf{x}_j)
    }_{\text{within-class purity terms}}
    -
    \underbrace{
    2
    \sum_{i \in I_+}
    \sum_{j \in I_-}
    k_{\Phi}(\mathbf{x}_i,\mathbf{x}_j)
    }_{\text{cross-class overlap term}}
    \Bigg]^{1/2}.
    \label{eq:app_ensemble}
\end{equation}

\begin{proof}  
Let
\[
A_{\Phi}
=
p_{+}\bar{\rho}_{+}
-
p_{-}\bar{\rho}_{-}.
\]
Since the Schatten norms satisfy
$\|A\|_{1}\ge \|A\|_{2}$ for any operator $A$, we have
\[
\Delta_{\Phi}
=
\|A_{\Phi}\|_{1}
\ge
\|A_{\Phi}\|_{2} = \|A_{\Phi}\|_{\mathrm{HS}},
\]
where
\[
\|A_{\Phi}\|_{\mathrm{HS}}
=
\sqrt{\operatorname{Tr}(A_{\Phi}^{\dagger}A_{\Phi})}.
\]
Because $A_{\Phi}$ is Hermitian,
\[
\|A_{\Phi}\|_{\mathrm{HS}}^{2}
=
\operatorname{Tr}(A_{\Phi}^{2}).
\]
Expanding the right-hand side gives
\[
\begin{aligned}
\operatorname{Tr}(A_{\Phi}^{2})
={}&
p_{+}^{2}\operatorname{Tr}(\bar{\rho}_{+}^{2})
+
p_{-}^{2}\operatorname{Tr}(\bar{\rho}_{-}^{2}) 
-
2p_{+}p_{-}
\operatorname{Tr}(\bar{\rho}_{+}\bar{\rho}_{-}).
\end{aligned}
\]

Using
\[
\bar{\rho}_{y}
=
\frac{1}{m_y}
\sum_{i\in I_y}\rho_{\Phi}(\mathbf{x}_i),
\qquad
p_y=\frac{m_y}{m},
\]
the class-purity terms can be written as
\[
\begin{aligned}
p_y^2\operatorname{Tr}(\bar{\rho}_y^2)
&=
\frac{p_y^2}{m_y^2}
\sum_{i,j\in I_y}
\operatorname{Tr}
\left[
\rho_{\Phi}(\mathbf{x}_i)\rho_{\Phi}(\mathbf{x}_j)
\right] \\
&=
\frac{1}{m^2}
\sum_{i,j\in I_y}
\operatorname{Tr}
\left[
\rho_{\Phi}(\mathbf{x}_i)\rho_{\Phi}(\mathbf{x}_j)
\right].
\end{aligned}
\]
Likewise, the cross-class term satisfies
\[
\begin{aligned}
p_{+}p_{-}
\operatorname{Tr}(\bar{\rho}_{+}\bar{\rho}_{-})
&=
\frac{p_{+}p_{-}}{m_{+}m_{-}}
\sum_{i\in I_{+}}
\sum_{j\in I_{-}}
\operatorname{Tr}
\left[
\rho_{\Phi}(\mathbf{x}_i)\rho_{\Phi}(\mathbf{x}_j)
\right] \\
&=
\frac{1}{m^2}
\sum_{i\in I_{+}}
\sum_{j\in I_{-}}
\operatorname{Tr}
\left[
\rho_{\Phi}(\mathbf{x}_i)\rho_{\Phi}(\mathbf{x}_j)
\right].
\end{aligned}
\]

Combining these expressions yields
\[
\Delta_{\Phi}
\ge
\frac{1}{m}
\left[
\sum_{y\in\{+,-\}}
\sum_{i,j\in I_y}
\operatorname{Tr}
\left[
\rho_{\Phi}(\mathbf{x}_i)\rho_{\Phi}(\mathbf{x}_j)
\right]
-
2
\sum_{i\in I_{+}}
\sum_{j\in I_{-}}
\operatorname{Tr}
\left[
\rho_{\Phi}(\mathbf{x}_i)\rho_{\Phi}(\mathbf{x}_j)
\right]
\right]^{1/2}.
\]

This derivation holds for arbitrary density operators
$\rho_{\Phi}(\mathbf{x}_i)$, including mixed states. For the pure-state
embeddings considered in this work,
\[
\rho_{\Phi}(\mathbf{x}_i)
=
|\psi_{\Phi}(\mathbf{x}_i)\rangle
\langle\psi_{\Phi}(\mathbf{x}_i)|,
\]
and therefore
\[
\operatorname{Tr}
\left[
\rho_{\Phi}(\mathbf{x}_i)\rho_{\Phi}(\mathbf{x}_j)
\right]
=
\left|
\langle
\psi_{\Phi}(\mathbf{x}_i)
\mid
\psi_{\Phi}(\mathbf{x}_j)
\rangle
\right|^2
=
k_{\Phi}(\mathbf{x}_i,\mathbf{x}_j),
\]
which gives Eq.~\ref{eq:app_ensemble}.
\end{proof}
The within-class sums include the diagonal terms $i=j$. For normalized pure-state embeddings, these terms satisfy $k_\Phi(\mathbf{x}_i,\mathbf{x}_i)=1$ and require no quantum evaluation. Together with the symmetry of the fidelity kernel, this leaves $B(B-1)/2$ distinct off-diagonal fidelity evaluations for $B$ samples.

\subsection{Proposition~\ref{prop1}}
\label{appendix:prop1_proof}
For arbitrary class prototypes $\mathbf{c}_{+}$ and $\mathbf{c}_{-}$, class-ensemble distinguishability $\Delta_{\Phi}$ satisfies
\begin{equation*}
\left\lVert p_{+}\bar{\rho}_{+} - p_{-}\bar{\rho}_{-} \right\rVert_1
\geq
\sqrt{1-4p_{+}p_{-}F_c} - 2 \left( p_{+}\bar r_{+}^{\mathrm{tr}} + p_{-}\bar r_{-}^{\mathrm{tr}} \right).
\end{equation*}
\begin{proof}
By the reverse triangle inequality,
\begin{align*}
\left\lVert p_{+}\bar{\rho}_{+} - p_{-}\bar{\rho}_{-} \right\rVert_{1}
&\geq
\left\lVert p_{+}\rho_{+}^{c} - p_{-}\rho_{-}^{c} \right\rVert_{1}
- p_{+} \left\lVert \bar{\rho}_{+}
- \rho_{+}^{c} \right\rVert_{1}
- p_{-} \left\lVert \bar{\rho}_{-}
- \rho_{-}^{c} \right\rVert_{1}.
\end{align*}
The convexity of the trace norm gives
\begin{equation*}
\left\lVert \bar{\rho}_{y} - \rho_{y}^{c} \right\rVert_{1}
\leq
\frac{1}{m_y} \sum_{i\in I_y} \left\lVert \rho_{\Phi}(\mathbf{x}_i)
- \rho_{y}^{c} \right\rVert_{1}
= 2\bar r_{y}^{\mathrm{tr}}.
\end{equation*}
Moreover, for the two pure prototype states,
\begin{equation*}
\left\lVert p_{+}\rho_{+}^{c} - p_{-}\rho_{-}^{c} \right\rVert_{1}
=\sqrt{
(p_{+}-p_{-})^{2}
+
4p_{+}p_{-}(1-F_c)
} 
= \sqrt{1-4p_{+}p_{-}F_c}.
\end{equation*}
\end{proof}

\subsection{Proposition~\ref{prop:multi}}
\label{appendix:prop2_proof}
For the class-wise prototype sets $\mathcal{P}_y=\{\mathbf{c}_{y,1},\ldots,\mathbf{c}_{y,K_y}\}$, the class-ensemble distinguishability $\Delta_{\Phi}$ satisfies
\begin{align*}
\left\|
p_{+}\bar\rho_{+}
-
p_{-}\bar\rho_{-}
\right\|_{1}
&\geq
\|p_+\tilde\rho_+^c-p_-\tilde\rho_-^c\|_1
-
2\sum_{y\in\{+,-\}}p_y\bar r_{y,K_y}^{\mathrm{tr}}.
\label{eq:multi_certificate}
\end{align*}
In the balanced case,
$$
\left\|
\frac{1}{2}\bar\rho_{+}
-
\frac{1}{2}\bar\rho_{-}
\right\|_{1}
\geq
\|\frac{1}{2}\tilde\rho_+^c-\frac{1}{2}\tilde\rho_-^c\|_1
- (\bar{r}^{\mathrm{tr}}_{+,K_+} + \bar{r}^{\mathrm{tr}}_{-,K_-})
$$
\begin{proof}
By reverse triangle inequality,
$$
\|p_+\bar\rho_+-p_-\bar\rho_-\|_1
\ge
\|p_+\tilde\rho_+^c-p_-\tilde\rho_-^c\|_1
-
p_+\|\bar\rho_+-\tilde\rho_+^c\|_1
-
p_-\|\bar\rho_- -\tilde\rho_-^c\|_1.
$$
The induced-weight identity gives
$$
\tilde\rho_y^c
=
\sum_{a=1}^{K_y}
\pi_{y,a}\rho_{y,a}^{c}
=
\frac1{m_y}\sum_{a=1}^{K_y}|I_{y,a}|\rho^c_{y,a}
=
\frac1{m_y}\sum_{i\in I_y}\rho_{y,a(i)}^c.
$$
Therefore,
$$
\bar\rho_y-\tilde\rho_y^c
=
\frac1{m_y}\sum_{i\in I_y}
\left(\rho_\Phi(\mathbf{x}_i)-\rho_{y,a(i)}^c\right).
$$
Using convexity,
\begin{align*}
\|\bar\rho_y-\tilde\rho_y^c\|_1
&{}= \left\| \frac1{m_y}\sum_{i\in I_y}
\left(\rho_\Phi(\mathbf{x}_i)-\rho_{y,a(i)}^c\right)  \right\|_1 \\ 
&\le
\frac1{m_y}\sum_{i\in I_y}\|\rho_\Phi(\mathbf{x}_i)-\rho_{y,a(i)}^c\|_1 \\
&=
2\bar r_{y,K_y}^{\mathrm{tr}}.
\end{align*}
\end{proof}

\paragraph{Connection to the prototype-guided objective.}
Applying the Hilbert--Schmidt relaxation from Appendix~\ref{app:hs_relaxation} to the prototype-mixture separation gives
\[
\|p_+\tilde{\rho}_+^c-p_-\tilde{\rho}_-^c\|_1
\ge
\|p_+\tilde{\rho}_+^c-p_-\tilde{\rho}_-^c\|_{\mathrm{HS}},
\]
where
\begin{align*}
\|p_+\tilde{\rho}_+^c-p_-\tilde{\rho}_-^c\|_{\mathrm{HS}}^2
&\quad =
\sum_{y\in\{+,-\}}p_y^2
\sum_{a,a'=1}^{K_y}
\pi_{y,a}\pi_{y,a'}
k_\Phi(\mathbf{c}_{y,a},\mathbf{c}_{y,a'})\\
&\qquad -
2p_+p_-
\sum_{a=1}^{K_+}\sum_{b=1}^{K_-}
\pi_{+,a}\pi_{-,b}
k_\Phi(\mathbf{c}_{+,a},\mathbf{c}_{-,b}).
\end{align*}
This expansion motivates high within-class prototype fidelity and low cross-class prototype fidelity. For pure-state embeddings,
$1-k_\Phi(\mathbf{x}_i,\mathbf{c}_{y,a(i)}) =D_{\mathrm{tr}}(\rho_\Phi(\mathbf{x}_i),\rho_{y,a(i)}^c)^2$,
which provides a squared-distance surrogate for the radius terms. Thus, $L_{\mathrm{sample}}$, $L_{\mathrm{in}}$, and $L_{\mathrm{out}}$ serve as unweighted geometric surrogates for the corresponding certificate terms.

\section{Multiple-Prototype RKHS Witness for the SVM Objective}
\label{appendix:multi_proto_witness}

We extend Proposition~\ref{prop:svm_witness} to the multiple-prototype setting introduced in Proposition~\ref{prop:multi}. For each class $y\in\{+,-\}$, define the prototype mean embedding in the RKHS $\mathcal{H}_{\Phi}$ as
\[
\mu_y
=
\sum_{a=1}^{K_y}
\pi_{y,a}
k_{\Phi}(\cdot,\mathbf{c}_{y,a}).
\]
The corresponding multiple-prototype RKHS classifier is
\begin{align}
h_{\mathrm{multi}}(\mathbf{x})
&=
\mu_{+}(\mathbf{x})-\mu_{-}(\mathbf{x})
\nonumber\\
&=
\sum_{a=1}^{K_+}
\pi_{+,a}k_{\Phi}(\mathbf{x},\mathbf{c}_{+,a})
-
\sum_{b=1}^{K_-}
\pi_{-,b}k_{\Phi}(\mathbf{x},\mathbf{c}_{-,b}).
\label{eq:multi_proto_classifier}
\end{align}
Thus, $h_{\mathrm{multi}}$ compares the similarity of a
sample to the empirical prototype mixtures of the two classes,
rather than to a single prototype from each class.

By the reproducing property,
$\mu_{+},\mu_{-}\in\mathcal{H}_{\Phi}$ and hence
$h_{\mathrm{multi}}\in\mathcal{H}_{\Phi}$.
Its squared RKHS norm is
\begin{align*}
\|h_{\mathrm{multi}}\|_{\mathcal{H}_{\Phi}}^2
&=
\|\mu_{+}-\mu_{-}\|_{\mathcal{H}_{\Phi}}^2
\nonumber\\
&=
\|\mu_{+}\|_{\mathcal{H}_{\Phi}}^2
+
\|\mu_{-}\|_{\mathcal{H}_{\Phi}}^2
-
2
\langle
\mu_{+},\mu_{-}
\rangle_{\mathcal{H}_{\Phi}}.
\end{align*}
Applying the reproducing property to each term gives
\begin{align}
\|h_{\mathrm{multi}}\|_{\mathcal{H}_{\Phi}}^2
={}&
\sum_{a,a'=1}^{K_+}
\pi_{+,a}\pi_{+,a'}
k_{\Phi}(\mathbf{c}_{+,a},\mathbf{c}_{+,a'})
\nonumber\\
&+
\sum_{b,b'=1}^{K_-}
\pi_{-,b}\pi_{-,b'}
k_{\Phi}(\mathbf{c}_{-,b},\mathbf{c}_{-,b'})
\nonumber\\
&-
2
\sum_{a=1}^{K_+}
\sum_{b=1}^{K_-}
\pi_{+,a}\pi_{-,b}
k_{\Phi}(\mathbf{c}_{+,a},\mathbf{c}_{-,b}).
\label{eq:multi_proto_norm}
\end{align}

This construction leads directly to the multiple-prototype
analogue of Proposition~\ref{prop:svm_witness}.

\begin{proposition}[Multiple-Prototype RKHS Witness]
\label{prop:multi_svm_witness}
For a fixed quantum embedding $\Phi$, let
$h_{\mathrm{multi}}$ be defined by
Eq.~\ref{eq:multi_proto_classifier}.
Then the optimal normalized soft-margin SVM objective satisfies
\begin{align}
\mathcal{J}_{\mathrm{SVM}}^{\star}(\Phi)
\leq
\inf_{\tau\geq0}
\Bigg\{
&
\frac{\lambda_{\mathrm{SVM}}}{2}
\tau^2
\|h_{\mathrm{multi}}\|_{\mathcal{H}_{\Phi}}^2
+
\frac{1}{m}
\sum_{i=1}^{m}
\left[
1-\tau y_i h_{\mathrm{multi}}(\mathbf{x}_i)
\right]_{+}
\Bigg\}.
\label{eq:multi_svm_witness}
\end{align}
\end{proposition}

\begin{proof}
By definition, the optimal normalized soft-margin SVM
objective is
\begin{equation}
\mathcal{J}_{\mathrm{SVM}}^{\star}(\Phi)
=
\inf_{\substack{
f\in\mathcal{H}_{\Phi}\\
b\in\mathbb{R}
}}
\left\{
\frac{\lambda_{\mathrm{SVM}}}{2}
\|f\|_{\mathcal{H}_{\Phi}}^2
+
\frac{1}{m}
\sum_{i=1}^{m}
\left[
1-y_i(f(\mathbf{x}_i)+b)
\right]_{+}
\right\}.
\nonumber
\end{equation}

Since $h_{\mathrm{multi}}\in\mathcal{H}_{\Phi}$,
for every $\tau\geq0$ the scaled function
\[
f_{\tau}
=
\tau h_{\mathrm{multi}}
\]
also belongs to $\mathcal{H}_{\Phi}$.
Therefore,
$(f_{\tau},0)$ is an admissible candidate in the SVM
optimization problem, and hence
\begin{align*}
\mathcal{J}_{\mathrm{SVM}}^{\star}(\Phi)
&\leq
\frac{\lambda_{\mathrm{SVM}}}{2}
\tau^2
\|h_{\mathrm{multi}}\|_{\mathcal{H}_{\Phi}}^2
+
\frac{1}{m}
\sum_{i=1}^{m}
\left[
1-\tau y_i h_{\mathrm{multi}}(\mathbf{x}_i)
\right]_{+}
\end{align*}
for every $\tau\geq0$.
Taking the infimum over $\tau\geq0$ yields
Eq.~\ref{eq:multi_svm_witness}.
\end{proof}

Evaluating $h_{\mathrm{multi}}(\mathbf{x}_i)$ requires the fidelity
between $x_i$ and every prototype from both classes.
For a sample of size $B$, the sample-dependent cost is therefore
\[
B(K_{+}+K_{-})
\]
fidelity evaluations.
The RKHS norm in Eq.~\ref{eq:multi_proto_norm} can be
computed using
\[
P_K
=
\binom{K_+}{2}
+
\binom{K_-}{2}
+
K_+K_-
\]
distinct prototype-to-prototype fidelities, since the diagonal
terms satisfy $k_{\Phi}(\mathbf{c}_{y,a},\mathbf{c}_{y,a})=1$.
Thus, the total number of distinct fidelity evaluations required
for the multiple-prototype witness is
\begin{equation*}
N_{\mathrm{fid}}^{\mathrm{multi\text{-}witness}}
=
B(K_{+}+K_{-})
+
\binom{K_+}{2}
+
\binom{K_-}{2}
+
K_+K_- = O(BK_\mathrm{tot} + K_\mathrm{tot}^2).
\label{eq:multi_witness_cost}
\end{equation*}
It remains linear in $B$ for a fixed number of
prototypes.

Finally, when $K_+=K_-=1$, the prototype weights are unity and
\[
\mu_{+}=k_{\Phi}(\cdot,\mathbf{c}_{+}),
\qquad
\mu_{-}=k_{\Phi}(\cdot,\mathbf{c}_{-}),
\]
so that
\[
h_{\mathrm{multi}}
=
k_{\Phi}(\cdot,\mathbf{c}_{+})
-
k_{\Phi}(\cdot,\mathbf{c}_{-})
=
h_{\mathrm{proto}}.
\]
Moreover,
\[
\|h_{\mathrm{multi}}\|_{\mathcal{H}_{\Phi}}^2
=
2
-
2k_{\Phi}(\mathbf{c}_{+},\mathbf{c}_{-})
=
2(1-F_c).
\]
Hence Proposition~\ref{prop:multi_svm_witness} reduces exactly
to Proposition~\ref{prop:svm_witness}.

\section{Dataset Details}
\label{appendix:data}

This section summarizes the datasets, label construction, preprocessing, and Train/validation/Test partitions used in our experiments.

\begin{table}[t]
\centering
\begin{minipage}{0.84\textwidth}
    \centering
    \caption{
    Summary of the datasets used in the experiments.
    The \# Features column reports the input dimensionality, and $\rightarrow$ indicates dimensionality reduction prior to quantum embedding.
    }
\end{minipage}
\label{tab:data_details}
\small
\begin{tabular}{lcccc}
\toprule
Dataset
& Samples
& \#.Features
& Qubits
& Train / Val / Test \\
\midrule
IRIS~\citep{iris}
& 150
& 4
& 4
& 84 / 21 / 45 \\

VERTEBRAL~\citep{vert}
& 310
& 6
& 6
& 173 / 44 / 93 \\

T-CELL~\citep{tcell}
& 246
& 60 $\rightarrow$ 8
& 8
& 137 / 35 / 74 \\

GLASS~\citep{glass}
& 213
& 9 
& 9
& 119 / 30 / 64 \\

T-MNIST~\citep{mnist}
& 550
& 784 $\rightarrow$ 20
& 20
& 150 / 150 / 250 \\
\bottomrule
\end{tabular}
\end{table}

\paragraph{Common split and preprocessing protocol.}
Common split and preprocessing protocol. For IRIS, VERTEBRAL, and GLASS, we first construct stratified 70/30 outer-training/Test splits and reserve 20\% of each outer-training partition for validation, resulting in approximately 56/14/30 Train/validation/Test proportions. T-CELL uses a fixed 172/74 outer split, with the 172-example training pool further divided into 137 Train and 35 validation examples. All competing methods use identical sample indices within each dataset and seed. However, T-MNIST uses a fixed 150/150/250 Train/validation/Test split.

\paragraph{IRIS} We use all 150 examples and four continuous features from the UCI Iris dataset~\citep{iris}, retaining the original three balanced species classes. A \textit{StandardScaler}~\citep{scikit-learn} fitted on the outer-training pool transforms the four inputs, which are mapped directly to four qubits without dimensionality reduction.

\paragraph{VERTEBRAL} We use the binary Vertebral Column 2C~\citep{vert} task containing 210 abnormal and 100 normal examples with six biomechanical features. After standardization, each coordinate is transformed as $\mathbf{x}\rightarrow (\pi/2)\tanh(\mathbf{x})$ and mapped one-to-one to six qubits.

\paragraph{T-CELL} We use 246 examples from the combinatorial CAR-T signaling-motif dataset of \citet{tcell}. Four motif positions are represented as 15-category one-hot vectors, yielding a 60-dimensional input. We label an example positive when its \texttt{Nalm 6 Cytotoxicity} exceeds $0.62$ and negative otherwise. The fixed 172-example training pool is stratified into 137 Train and 35 validation examples, with 74 additional examples used for Test. A CNN--LSTM encoder~\citep{tcell} trained on the Train subset maps the inputs to eight angles via $\mathbf{z}=\pi\tanh(\mathbf{z}_{\mathrm{raw}})$. The frozen encoder is applied to validation and Test, and the eight coordinates are mapped directly to eight qubits.

\paragraph{GLASS} We use the nine continuous chemical features of the UCI Glass Identification dataset~\citep{glass} after removing the identifier and one duplicated feature vector, resulting in 213 examples. The original labels are regrouped into three classes: float-processed window glass, non-float-processed window glass, and non-window glass. Features are standardized and transformed by $\mathbf{x}\rightarrow (\pi/2)\tanh(\mathbf{x})$ before being mapped to nine qubits.

\paragraph{T-MNIST}We construct a balanced five-class subset of MNIST~\citep{mnist} using digits 0--4. We select 30 Train, 30 validation, and 50 Test images per digit, giving fixed partitions of 150/150/250 examples, and scale pixel values to $[0,1]$. A convolutional Wasserstein autoencoder (WAE)~\citep{wae} trained on the 150 train images produces a 20-dimensional latent representation. The latents are standardized using statistics of trainset, transformed as $\mathbf{x}\mapsto(\pi/2)\tanh(\mathbf{x})$, and mapped one-to-one to 20 qubits. The learned encoder and transformation are applied unchanged to validation and Test.

\section{Experimental Details}
\label{appendix:experiments_details}

\subsection{Genetic-Algorithm Circuit Search}
\label{appendix:ga-search}

\paragraph{Stage-wise circuit construction.}
Let $Q$ and $d$ denote the number of qubits and the input dimension, respectively. We construct a data-embedding circuit with at most $L_{\max}=5$ encoding layers,
\[
    U_{L}(\mathbf{x})=B_{L}(\mathbf{x})\cdots B_{1}(\mathbf{x}),
    \qquad
    B_{\ell}(\mathbf{x})=E_{\mathrm{lat}}V_{\ell}(\mathbf{x}).
\]
Each one-qubit block $V_\ell$ contains two data-dependent rotations per qubit,
\[
    V_{\ell}(\mathbf{x})
    =
    \bigotimes_{q=1}^{Q}
    \bigg[R_{a_{\ell q2}}\!\left(\tau_{\ell q2}(\mathbf{x})\right)
    R_{a_{\ell q1}}\!\left(\tau_{\ell q1}(\mathbf{x})\right)\bigg],
\]
where the rotation axis is selected from
\[
    a_{\ell qs}\in\{X,Y,Z\},
\]
and angle function is selected from
\[\mathcal T(\mathbf{x})
    =
    \{\mathbf{x}_j:1\le j\le d\} 
    \cup
    \left\{\frac{\mathbf{x}_j\mathbf{x}_k}{\pi}:1\le j<k\le d\right\}
    \cup\{0,\pi\},
    \qquad
    \tau_{\ell qs}(\mathbf{x})\in\mathcal T(\mathbf{x}).\]
Thus, each one-qubit gene specifies a rotation axis together with a single-feature, feature-product, or constant rotation angle. After crossover and mutation, a repair operation ensures that every input feature appears at least once among the rotation slots of each active layer.

The two-qubit block $E_\mathrm{lat}$ is fixed as
\[
    E_{\mathrm{lat}}
    =
    \prod_{(u,v)\in\mathcal E_{\mathrm{lat}}}
    \mathrm{CZ}_{uv},
\]
where $\mathcal E_{\mathrm{lat}}$ contains the horizontal and vertical nearest-neighbor edges of an open rectangular lattice. This connectivity was motivated by the square-lattice architecture of the recent IBM Quantum Nighthawk processor~\citep{ibm_nighthawk}. We therefore use a fixed nearest-neighbor lattice throughout the search and optimize only the data-encoding single-qubit rotations.

\begin{table*}[t]
    \centering
    \begin{minipage}{0.8\textwidth}
    \caption{Dataset-specific circuit layouts and GA hyperparameters. $Q$ is the number of qubits, lattice specifies the rows and columns of the rectangular grid, and $|\mathcal E_{\mathrm{lat}}|$ counts its nearest-neighbor CZ edges per layer. Population and Epochs/stage denote the population size and maximum number of generations at each of the five sequential stages. $E_{\mathrm{keep}}$ is the number of tournament-selected parents protected from mutation. Patience is the number of consecutive generations without sufficient improvement in the stage's monitored criterion before early stopping.}
    \label{tab:ga-hyperparameters}
    \end{minipage}
    \small
    \setlength{\tabcolsep}{4.5pt}
    \begin{tabular}{lccccccc}
        \toprule
        Dataset & $Q$ & Lattice & $|\mathcal E_{\mathrm{lat}}|$
        & Population & Epochs/stage & $E_{\rm keep}$ & Patience \\
        \midrule
        IRIS       & 4  & $2\times2$ & 4  & 200 & 50  & 8 & 10 \\
        VERTEBRAL  & 6  & $2\times3$ & 7  & 300 & 100 & 5 & 15 \\
        T-CELL     & 8  & $4\times2$ & 10 & 300 & 100 & 5 & 15 \\
        GLASS      & 9  & $3\times3$ & 12 & 300 & 100 & 5 & 15 \\
        T-MNIST    & 20 & $4\times5$ & 31 & 300 & 100 & 5 & 15 \\
        \bottomrule
    \end{tabular}
\end{table*}

\paragraph{Sequential optimization and candidate-bank construction.}
At stage $\ell$, the previously selected layers $B_1,\ldots,B_{\ell-1}$ form a frozen prefix, and the GA evolves only one newly initialized active layer $B_l$. Each individual is concatenated with this prefix and evaluated using the corresponding training objective. The population is updated by tournament selection with $k=4$, structural crossover, and mutation with probability $0.15$. Conditional on mutation, the active genome is re-randomized either at the whole-layer or single-qubit level. The best training-loss individual is appended to the frozen prefix at the end of each stage. Early stopping is applied independently when the incumbent fails to improve by at least $10^{-4}$ for the patience reported in Table~\ref{tab:ga-hyperparameters}.

We retain a global bank of the 20 unique circuits with the lowest observed training losses among candidates of depths 2–5. Candidate identity is defined by the complete circuit genome. Validation and test data are not used to evolve the GA population or construct the candidate bank. After the search, a prespecified validation rule selects the final circuit, and the selected circuit is fixed during the NQE fine-tuning stage. 

\paragraph{Fidelity kernel and prototype construction.}
For a circuit genome $G$ and its embedding unitary $U_\Phi$, we use the fidelity kernel
\[
    k_{\Phi}(\mathbf{x},\mathbf{x}')
    =
    \left|
        \langle 0|U_{\Phi}^{\dagger}(\mathbf{x})U_{\Phi}(\mathbf{x}')|0\rangle
    \right|^2.
\]
Let $N$ be the training size, $C$ the number of classes, and $n_y$ the number of training samples in class $y$, and $P=\binom{C}{2}$.

Because the input features are used as rotation-gate angles, each coordinate can be interpreted as a periodic angular variable. We therefore construct the $K=1$ class prototype in the input angle space using the coordinate-wise circular mean rather than the arithmetic mean,
\[
    (\mathbf{c}_y)_j
    =
    \operatorname{atan2}\!\left(
        \frac{1}{n_y}\sum_{i:y_i=y}\sin \mathbf{x}_{ij},
        \frac{1}{n_y}\sum_{i:y_i=y}\cos \mathbf{x}_{ij}
    \right).
\]

\paragraph{Prototype objective.}
The single Prototype GA minimizes
\[
    \mathcal L_{\mathrm P}(\Phi)
    =
    \underbrace{\frac{1}{\binom{C}{2}}
    \sum_{a<b}k_{\Phi}(\mathbf{c}_a,\mathbf{c}_b)
    }_{\mathcal L_{\mathrm{cross}}}
    +
    \lambda_{\mathrm P}
    \underbrace{\frac{1}{C}
    \sum_y \frac{1}{n_y}
    \sum_{i:y_i=y}
    \left[1-k_{\Phi}(\mathbf{x}_i,\mathbf{c}_y)\right]
    }_{\mathcal L_{\mathrm{sample}}},
\]
The first term penalizes fidelity between distinct class prototypes, while the second promotes high fidelity between each sample and its own class prototype. The weight $\lambda_{\mathrm P}$ is selected using validation data at the dataset level and remains fixed during test evaluation. The values used for IRIS, VERTEBRAL, T-CELL, GLASS, and T-MNIST are $0.5$, $0.7$, $0.3$, $0.5$, and $0.5$, respectively.

For $N$ training samples and $C$ classes, one evaluation of the prototype objective requires
\[
    Q_{\mathrm P}=N+\binom{C}{2}
\]
logical fidelity evaluations.

\paragraph{Multi-prototype objective.}
For GLASS, we extend the single-prototype construction by retaining one prototype per original sublabel within each merged class, giving $(K_1,K_2,K_3)=(2,1,3)$. Prototypes are computed as coordinatewise circular means of their inner-Train samples. Each sample remains assigned to its original sublabel prototype, and $\pi_{y,a}=n_{y,a}/n_y$ denotes its fixed within-class mixture weight.

We extend $\mathcal L_{\mathrm{sample}}$ by replacing $c_y$ with the assigned prototype $\mathbf{c}_{y,a(i)}$. For $\mathcal L_{\mathrm{cross}}$, each class-pair fidelity is replaced by the mixture-weighted fidelity $\sum_{a,b}\pi_{y,a}\pi_{z,b}k_\Phi(\mathbf{c}_{y,a},\mathbf{c}_{z,b})$. We additionally penalize separation between prototypes belonging to the same class,
\begin{equation*}
\mathcal L_{\mathrm{in}}(\Phi)
=
\frac{1}{C^2}
\sum_y
\sum_{a,a'=1}^{K_y}
\pi_{y,a}\pi_{y,a'}
\left[1-k_\Phi(\mathbf{c}_{y,a},\mathbf{c}_{y,a'})\right].
\end{equation*}

Unlike the single-prototype weighted sum, Multi-$K$ GA jointly optimizes three components using an augmented Tchebycheff fitness. The implemented component vector is
\begin{equation*}
\bigl(\ell_1,\ell_2,\ell_3\bigr)
=
\left(
\mathcal L_{\mathrm{sample}},
\mathcal L_{\mathrm{in}},
\frac{C-1}{C}\mathcal L_{\mathrm{cross}}
\right).
\end{equation*}
Here, the cross-class coefficient is $2/C^2$, rather than $1/\binom{C}{2}$ used in the single-prototype objective.

For each component, we define
\begin{equation*}
z_j(\Phi)
=
\frac{\ell_j(\Phi)-q_{j,0.1}}
{\max\{q_{j,0.9}-q_{j,0.1},10^{-3}\}},
\end{equation*}
where $q_{j,0.1}$ and $q_{j,0.9}$ are its empirical 10th and 90th percentiles in the initial population of the first GA stage. The final scalar fitness is
\begin{equation*}
\mathcal L_{\mathrm{MK}}(\Phi)
=
\max_{j\in\mathcal A}z_j(\Phi)
+
\frac{0.05}{|\mathcal A|}
\sum_{j\in\mathcal A}z_j(\Phi),
\end{equation*}
where $\mathcal A$ contains components whose initial interquantile range exceeds $10^{-12}$. Thus, GA primarily minimizes the largest normalized component, with a small additional penalty on their mean. All three component coefficients are one, and the normalization references remain fixed throughout the search. Each fitness evaluation requires $N+\binom{K_{\mathrm{tot}}}{2}=N+15$ fidelity terms.

\paragraph{Full-pairwise objective.}
Let
\[
    \mathcal P_{aa}
    =\{(i,j):i<j,\;y_i=y_j=a\},
    \qquad
    \mathcal P_{ab}
    =\{(i,j):y_i=a,\;y_j=b\},\quad a<b.
\]
The Full-pairwise GA (\textit{Full}) uses the class-balanced loss
\[
    \mathcal L_{\mathrm F}(\Phi)
    =
    \frac{1}{2}\mathcal L_{\mathrm{same}}(\Phi)
    +
    \frac{1}{2}\mathcal L_{\mathrm{diff}}(\Phi),
\]
where
\[
    \mathcal L_{\mathrm{same}}(\Phi)
    =
    \frac{1}{C}\sum_a
    \frac{1}{|\mathcal P_{aa}|}
    \sum_{(i,j)\in\mathcal P_{aa}}
    \left[1-k_{\Phi}(\mathbf{x}_i,\mathbf{x}_j)\right]^2,
\]
and
\[
    \mathcal L_{\mathrm{diff}}(\Phi)
    =
    \frac{1}{P}\sum_{a<b}
    \frac{1}{|\mathcal P_{ab}|}
    \sum_{(i,j)\in\mathcal P_{ab}}
    k_{\Phi}(\mathbf{x}_i,\mathbf{x}_j)^2.
\]
Because every unordered training pair is evaluated, its logical fidelity cost per candidate is
\[
    Q_{\mathrm F}=\binom{N}{2}.
\]

\paragraph{Budget-matched objective.}
To separate the effect of prototype-based supervision from the number of fidelity evaluations, the budget-matched baseline (\textit{Pairwise-B}) uses the same logical budget as the Prototype objective. At generation $t$, a deterministic seeded set $\mathcal E_t$ of
\[
    Q_{\mathrm U}=N+\binom{C}{2}=Q_{\mathrm P}
\]
training pairs is sampled without replacement from the complete pair pool. The same sampled set is used for all candidate circuits in that generation and is refreshed at the next generation.

Writing $\mathcal E_t^{\mathrm{same}}$ and $\mathcal E_t^{\mathrm{diff}}$ for its same- and different-class subsets, the objective is
\[
    \mathcal L_{\mathrm U}^{(t)}(\Phi)
    =
    \frac{1}{2|\mathcal E_t^{\mathrm{same}}|}
    \sum_{(i,j)\in\mathcal E_t^{\mathrm{same}}}
    \left[1-k_{\Phi}(\mathbf{x}_i,\mathbf{x}_j)\right]^2
    +
    \frac{1}{2|\mathcal E_t^{\mathrm{diff}}|}
    \sum_{(i,j)\in\mathcal E_t^{\mathrm{diff}}}
    k_{\Phi}(\mathbf{x}_i,\mathbf{x}_j)^2.
\]
\textit{Prototype} and \textit{Pairwise-B} therefore use identical fidelity budgets while differing in their supervision structure. The \textit{Prototype} objective therefore connects each sample to a persistent class anchor, whereas \textit{Pairwise-B} observes a refreshed subset of sample pairs.

\textit{Prototype} and \textit{Full} use the generation-best raw training loss to update the stage incumbent. \textit{Pairwise-B} uses an EMA of the generation-best sampled loss, with decay 0.9, and retains the raw-loss-best individual from the corresponding generation. Candidate-bank ranking always uses observed raw losses.

\paragraph{Validation-based circuit selection.}
After GA search, we select the circuit with the lowest validation loss from the frozen candidate bank. \textit{Prototype }uses the class-balanced prototype objective, whereas the \textit{Full} use their corresponding validation pairwise losses. For the \textit{Pairwise-B}, the validation pair set is fixed across candidate circuits. The selected circuit is then kept fixed during NQE refinement. Test data are not used for circuit selection.

\subsection{Neural Quantum Embedding Refinement}
\label{app:nqe}

\paragraph{Fixed-circuit residual refinement.}
After the GA selects embedding circuit $\Phi$, NQE keeps all quantum gates, feature assignments, and entangling edges fixed and learns only a classical residual map before the circuit. For every dataset and method, $g_\theta$ is the same shared simple D2--H32 network,
\begin{align*}
    r_\theta(\mathbf{x})
    &=W_3\,\operatorname{ReLU}\!\left(
        W_2\,\operatorname{ReLU}(W_1\mathbf{x}+b_1)+b_2
      \right)+b_3,\label{eq:nqe-residual-branch}\\
    g_\theta(\mathbf{x})
    &=\mathbf{x}+\widetilde r_\theta(\mathbf{x}).
\end{align*}
Both hidden layers have width 32. The weight and bias of output layer $(W_3, b_3)$ are initialized to zero, so $g_{\theta_0}(\mathbf{x})=\mathbf{x}$ exactly and epoch zero reproduces the selected GA embedding. This identity initialization prevents a random input transformation from perturbing the GA-optimized embedding before refinement begins. A single mapper transforms the input, and the resulting features are shared across all encoding layers of the fixed circuit.

We compare two residual domains while keeping the neural network architecture fixed. The unbounded version uses $\widetilde r_\theta=r_\theta$, whereas the bounded version uses
\begin{equation*}
    \widetilde r_\theta(x)
    =
    \alpha\tanh\!\left(\frac{r_\theta(x)}{\alpha}\right),
    \qquad \alpha=\frac{\pi}{4},
\end{equation*}
where the hyperbolic tangent is applied coordinate-wise. Hence the latter enforces $|g_\theta(\mathbf{x})_j-\mathbf{x}_j|\leq\pi/4$ without changing the derivative of the residual branch at the identity initialization. This comparison examines how restricting the magnitude of the learned input corrections affects NQE fine-tuning and downstream classification performance.

\paragraph{Training objectives.}
During NQE refinement, we optimize only the residual network parameters $\theta$, keeping the GA-selected circuit fixed. The prototypes remain fixed in input space, and both samples and prototypes are passed through the same residual mapper before fidelity evaluation. \textit{Prototype} and \textit{Pairwise-B} NQE retain their respective GA training objectives, evaluated on the mapped inputs.

Improving average prototype compactness does not necessarily improve the samples near the downstream decision boundary, as illustrated in Appendix~\ref{app:t_mnist_bad}. This motivates examining whether continued compactness minimization is beneficial during NQE refinement. We therefore consider a Dual-Guard variant that removes the direct incentive for further compactness reduction once a prescribed class-wise improvement has been attained. Specifically, we replace the mean compactness penalty with a soft constraint penalty relative to the initial GA embedding.

Let
\begin{equation*}
Q_y(\theta)
=
\frac{1}{n_y}\sum_{i:y_i=y}
\left[1-k_{\Phi,\theta}(\mathbf{x}_i,\mathbf{c}_y)\right],
\qquad
v_y(\theta)
=
\frac{Q_y(\theta)}
{\max\{Q_y(\theta_0),10^{-8}\}}
-(1-\delta).
\end{equation*}
Here, $\theta_0$ denotes the identity-initialized mapper, and
$v_y(\theta)\leq 0$ indicates that the prescribed compactness
target has been attained.
The guarded objective is
\begin{equation*}
\mathcal L_{\mathrm G}(\theta,\boldsymbol{\eta})
=
\mathcal L_{\mathrm{cross}}(\Phi,\theta)
+
\frac{1}{C}\sum_y
\left(
\eta_y[v_y(\theta)]_+
+
\frac{\rho}{2}[v_y(\theta)]_+^2
\right),
\end{equation*}
where $[u]_+=\max\{u,0\}$, $\delta=0.05$, and $\rho=0.1$.
The multipliers are initialized as
$\eta_y^{(0)}=\lambda_{\mathrm P}Q_y(\theta_0)$
and updated by projected dual ascent:
\begin{equation*}
\eta_y^{(t+1)}
=
\operatorname{clip}_{[0,10]}
\left(
\eta_y^{(t)}+0.05\,v_y(\theta_t)
\right).
\end{equation*}
This soft penalty encourages, but does not guarantee, satisfaction of the class-wise targets. It reuses the same sample--prototype fidelities and therefore requires no additional quantum evaluations.

\paragraph{Optimization and hyperparameter grid.}
Each run performs 40 Adam updates with a constant learning rate, zero weight decay, and global gradient-norm clipping at 1.0. Prototype uses the full inner-training sample--prototype objective, whereas \textit{Pairwise-B} uses a refreshed budget-matched pair set at each update. For Dual Guard, $\lambda_{\mathrm P}$ determines the initial dual multipliers rather than a fixed compactness weight. The residual mapper states are saved at
\begin{equation*}
    \mathcal T=\{0,5,10,15,20,25,30,35,40\}.
\end{equation*}
For the four tabular datasets, the validation search uses
\begin{align*}
    \alpha
    \in\left\{\frac{\pi}{4},\ \text{unbounded}\right\},\quad
    \eta_{\theta}
    \in\{0.001,0.003,0.005,0.01\}.
\end{align*}
The D2--H32 architecture, optimizer, training horizon, and all remaining settings are held common across this grid. The bound and learning rate are frozen at the dataset--method level using validation data only; Test results do not enter this choice. Table~\ref{tab:nqe-final-settings} lists the frozen endpoints.

\begin{table*}[t]
    \centering
    \begin{minipage}{0.78\textwidth}
    \caption{Frozen NQE hyperparameters used for the main comparison.``Unb.'' denotes an unbounded residual. Checkpoints are selected using each method's validation objective.}
    \label{tab:nqe-final-settings}
    \end{minipage}
    \small
    \setlength{\tabcolsep}{4.0pt}
    \begin{tabular}{lcccc}
        \toprule
        Dataset & Prototype loss & $\lambda_{\mathrm P}$
        & Prototype bound/LR &  Pairwise-B bound/LR \\
        \midrule
        IRIS      & Mean       & 0.5 & Unb./0.010     & Unb./0.005 \\
        VERTEBRAL & Mean       & 0.7 & Unb./0.010     & Unb./0.010 \\
        T-CELL    & Mean       & 0.3 & Unb./0.003     & Unb./0.010 \\
        GLASS     & Dual Guard & 0.5 & $\pi/4$/0.010  & Unb./0.010 \\
        T-MNIST   & Mean       & 0.5 & $\pi/4$/0.001  & $\pi/4$/0.001 \\
        \bottomrule
    \end{tabular}
\end{table*}

\paragraph{Checkpoint and QKSVM selection.}
NQE is trained only on the train partition. For Mean Prototype, the checkpoint in $\mathcal T$ with the smallest validation loss is selected using the frozen inner-Train prototypes. Dual Guard instead uses a feasibility-first rule. Checkpoints are ranked by the maximum positive class-wise validation violation, followed by the mean positive violation, cross-prototype fidelity, and the earlier epoch. Validation violations are measured relative to the epoch-zero validation compactness.

For \textit{Pairwise-B}, the checkpoint with the smallest validation pairwise loss is selected. The same seeded validation-pair set is used across checkpoints for the tabular datasets, whereas T-MNIST uses all unordered validation pairs. These validation evaluations do not update the mapper.

After the NQE refinement, the kernel on the complete outer-training set, comprising inner Train and validation, is constructed. For the single-prototype experiments and their \textit{Pairwise-B} baselines, the QKSVM regularization parameter is selected by stratified five-fold cross-validation on this set, using MCC and the grid
\begin{equation}
\mathcal C_{\mathrm{SVM}}
=
\{0.01,0.05,0.1,0.25,0.5,1,2.5,5,10,25,50\}.
\label{svm_grid}
\end{equation}
The QKSVM is then refitted on the complete outer-training set and evaluated on the held-out Test set. The mapper is not refitted on validation or Test data, and neither the checkpoint nor the SVM regularization parameter is reselected using Test performance.

\paragraph{Multi-prototype GLASS refinement.}
Multi-$K$ NQE retains the GA-stage prototypes, sublabel assignments, and mixture weights with $(K_1,K_2,K_3)=(2,1,3)$. The mapper architecture and training schedule are unchanged. Mean NQE minimizes $\ell_1+\ell_2+\ell_3$ without Tchebycheff normalization, whereas Dual Guard replaces $\ell_1$ with the class-wise guard, initialized with $\eta_y^{(0)}=Q_y(\theta_0)$. Checkpoint selection follows the corresponding Mean or Guard rule, with $\ell_2+\ell_3$ as the Guard structural tie-breaker.

The selected bound/LR settings are unbounded/$0.01$ for Mean and $(\pi/4)/0.003$ for Dual Guard. QKSVM evaluation follows the same procedure and the comparison also uses independently validation-selected single- and multi-prototype endpoints, with no Test-based reselection.

\section{Validation-Based Circuit Selection}
\label{app:selection}

\paragraph{Common candidate bank.}
The selector comparison changes only the post-search choice of circuit. For each dataset and seed, we use the same bank of $M=20$ unique prototype-optimized circuits as the main GA comparison. The bank retains the lowest training-loss candidates encountered across GA stages after genotype deduplication. Validation scores do not affect GA fitness or bank membership. We consider the single-prototype setting, with one fixed training
prototype per class. The train and validation sizes are denoted by $N$ and $V$, respectively, and the number of classes by $C$. Class prototypes $\mathbf{c}_y$ are computed once from train using the same coordinate-wise circular means as in GA search and remain fixed in input space for all candidates. The three selectors below rerank this identical bank.

\paragraph{P0: validation prototype loss.}
P0 selects the candidate with the smallest GA prototype objective evaluated on the validation set, keeping the training prototypes and $\lambda_{\mathrm P}$ fixed. Each candidate requires $V+\binom{C}{2}$ fidelity evaluations, one per validation sample--own-prototype pair and one per unordered prototype pair. Ties are resolved by the training-loss ordering of the bank and then by the deterministic candidate identifier.

\paragraph{P1: exact prototype RKHS classifier.}
P1 evaluates each validation sample against every class prototype and predicts
\begin{equation*}
\widehat y_\Phi(x)
=\mathop{\mathrm{arg\,max}}_{y\in\{1,\ldots,C\}} k_\Phi(x,c_y).
\end{equation*}
For two classes, this is the sign rule induced by $h_{ab}(x)=k_\Phi(x,c_a)-k_\Phi(x,c_b)$, the prototype RKHS witness. The multiclass implementation selects the largest class score directly.

Candidates are ranked by validation MCC of these predictions. To resolve MCC ties, we use the 10th percentile of the normalized prototype margins,
\begin{equation*}
\gamma_i^{\mathrm{proto}}(\Phi)
=\min_{b\ne y_i}
\frac{k_\Phi(x_i,c_{y_i})-k_\Phi(x_i,c_b)}
{\sqrt{2\bigl[1-k_\Phi(c_{y_i},c_b)\bigr]}}.
\end{equation*}
The denominator is $\lVert h_{y_i b}\rVert_{\mathcal H_\Phi}$. The margin therefore measures the least favorable true-class comparison in RKHS units. When candidates have the same validation MCC, we prefer the one with a larger margin at the 10th percentile of the validation margin distribution, favoring better separation near the difficult end of the validation samples. If this value is also tied, we prefer the smaller validation prototype loss and then use the fixed candidate order. Candidates whose prototypes from different classes map to nearly identical quantum states ($\sqrt{1-k_\Phi(c_a,c_b)}\leq10^{-12}$) are given lower priority, because their normalized margins involve division by a near-zero RKHS norm.

P1 requires neither SVM fitting nor a regularization search. In particular, it uses the witness classifier's validation MCC, not a direct minimization of the SVM-objective upper bound in Proposition~\ref{prop:svm_witness} or an optimization of its scaling parameter.

\paragraph{P2: validation QKSVM reference.}
For every candidate, P2 constructs the Gram matrix of train set and chooses the SVM regularization parameter by stratified five-fold MCC cross-validation on the train set. The QKSVM is then fitted on all train samples and scored on validation samples using the rectangular validation--train kernel whose size is $V\times N$. The candidate with the largest validation MCC is selected; ties prefer larger inner-Train CV MCC, smaller validation prototype loss. We set the SVM regularization grid in Eq.~\ref{svm_grid} and the CV folds are shared across candidates within each split. P2 is therefore a sample-trained QKSVM reference, not another prototype-classifier rule.

\paragraph{Validation cost.}
For each candidate circuit, the numbers of distinct fidelity evaluations required by the selectors are
\begin{align*}
Q_{\mathrm{P0}}&=V+\binom{C}{2},\\
Q_{\mathrm{P1}}&=CV+\binom{C}{2},\\
Q_{\mathrm{P2,\,SVM}}&=\binom{N}{2}+NV.
\end{align*}
The prototype-pair evaluations in P1 are used to normalize the margins for tie-breaking. If P0 evaluations are available, P1 requires only $(C-1)V$ additional sample--prototype fidelities. Self-fidelities equal one and need not be evaluated, and the same kernels are reused across SVM hyperparameters and CV folds. These counts describe logical fidelity evaluations rather than wall-clock runtime. For binary classification, P1 uses $2V$ sample--prototype evaluations and the corresponding multiclass cost is $CV$, in addition to the prototype-pair terms. P0 and P1 remain linear in $V$ for fixed $C$, whereas P2 requires the dense training kernel.

After circuit selection, all three selectors use the same downstream QKSVM evaluation protocol described in Appendix~\ref{appendix:experiments_details}. The comparison therefore assesses the prototype classifier as an embedding-selection proxy, not as a replacement for the fitted QKSVM or a guarantee of its Test performance.

\section{T-MNIST Optimization Analysis}
\label{app:tmnist}
\subsection{Why Prototype-Guided Optimization Performs well}
\label{app:t_mnist_good}

\begin{figure*}[t]
  \centering
  \includegraphics[width=\textwidth]{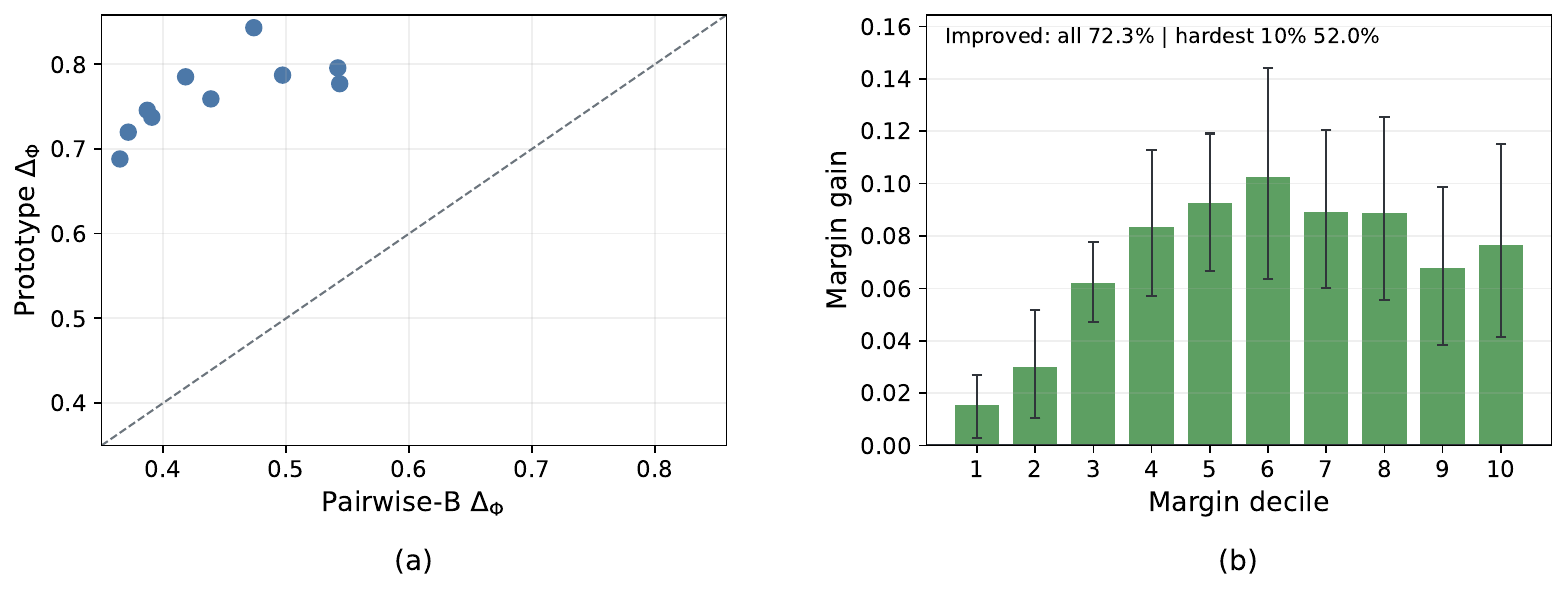}
  \begin{minipage}{1\linewidth}
        \centering
  \caption{
  Geometric analysis for the T-MNIST NQE endpoints in Table~\ref{tab:main_results}. \textbf{(a)} Exact train class-pair ensemble distance $\Delta_{\Phi,ab}$. Each marker represents one of the ten class pairs and reports its average over ten seeds.
  \textbf{(b)} Mean Test margin gain,
  $\Delta\widetilde m=\widetilde m^{\mathrm{Prototype}} -\widetilde m^{\mathrm{Pairwise\text{-}B}}$, 
  within deciles defined by the leave-one-seed-out \textit{Pairwise-B} difficulty. Decile 1 contains the hardest 10\% of Test samples. Error bars indicate 95\% seed-bootstrap intervals.
  }
  \label{fig:tmnist-nqe-advantage}
  \end{minipage}
\end{figure*}

We analyze the final T-MNIST embeddings reported in Table~\ref{tab:main_results}. For classes $a$ and $b$, let
\begin{equation}
    \Delta_{\Phi,ab}
    = \frac{1}{2}\left\|\rho_{\Phi,a}-\rho_{\Phi,b}\right\|_1,
    \qquad
    \rho_{\Phi,a}
    = \frac{1}{n_a}\sum_{i:y_i=a}
      \rho_\Phi(\mathbf{x}_i).
    \label{eq:tmnist-pair-ensemble-distance}
\end{equation}
We evaluate this quantity exactly on the train set for all class pairs with the final selected embedding model after NQE.

To connect this Train-side geometry to the downstream classifier, we also measure the weakest one-vs-one (OvO) QKSVM margin of each Test sample. Because the SVM regularization parameter can differ between methods and seeds, we use the RKHS-normalized geometric margin
\begin{equation}
    \widetilde m_i
    = \min_{k\ne y_i}
      \frac{d_{y_i,k}(\mathbf{x}_i)}{\lVert w_{y_i,k}\rVert_{\mathcal H}},
    \label{eq:tmnist-normalized-ovo-margin}
\end{equation}
where $d_{y_i,k}(\mathbf{x}_i)>0$ favors the true class in the binary comparison with class $k$. For each seed, we estimate the difficulty of a test sample using its average \textit{Pairwise-B} margin over the other nine seeds. We rank the samples by this independent reference and divide them into ten equal-sized groups, with decile 1 containing the hardest samples. Within each group, we compare the \textit{Prototype} and \textit{Pairwise-B} margins for the seed. Thus, the margin used to define sample difficulty does not come from the seed whose two methods are being compared.

Figure~\ref{fig:tmnist-nqe-advantage}(a) shows that all ten class-pair markers lie above the diagonal. Thus, \textit{Prototype} increases $\Delta_{\Phi,ab}$ for every class pair, rather than improving only the macro average. The average of exact class-ensemble distance increases from $0.443$ to $0.764$, accompanied by an increase in Test MCC from $0.642$ to $0.779$.

Figure~\ref{fig:tmnist-nqe-advantage}(b) further shows how this improved class geometry is reflected in the decision boundaries. \textit{Prototype} increases the geometric margin in $72.3\%$ of the observations, and the mean margin gain is positive in all ten difficulty deciles. However, the gain is weakest in the hardest decile. Only $52\%$ of its observations improve, and its mean gain is $0.0154$ with a 95\% interval of $[+0.0027,0.0270]$. Therefore, \textit{Prototype} improves class-ensemble separation and most test margins, but its advantage becomes less reliable near the most difficult decision boundaries. These findings provide a geometric account consistent with the \textit{Prototype} advantage reported in Table~\ref{tab:main_results}.

\subsection{Where Prototype-Guided Optimization Falls Short}
\label{app:t_mnist_bad}

\begin{figure}[t]
    \centering
    \includegraphics[width=1\linewidth]{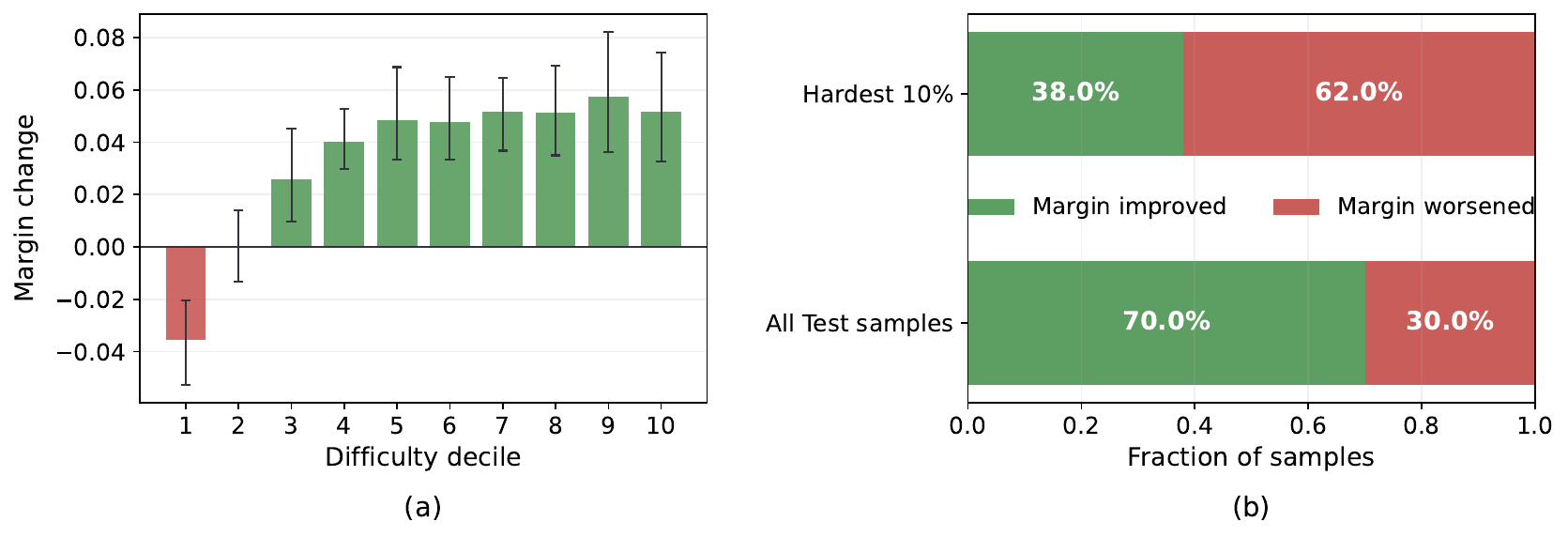}
        \centering
        \caption{T-MNIST lower-margin-tail diagnostic using the all-depth circuit bank from Figure~\ref{fig:certificate}.
        \textbf{(a)} Mean paired change in RKHS-normalized geometric margin from High- to Low-loss circuits across Test-sample difficulty deciles. Decile 1 contains the hardest samples, as determined by the leave-one-seed-out High-loss reference. Error bars denote 95\% seed-bootstrap intervals.
        \textbf{(b)} Fractions of paired seed--sample observations whose margin improves or worsens, evaluated over all Test samples and over the hardest decile. Circuit groups and QKSVM hyperparameters are fixed before Test evaluation.}
        \label{fig:app}
\end{figure}

We further examine why the favorable geometry produced by the prototype objective on T-MNIST does not translate monotonically into improved QKSVM performance in Figure~\ref{fig:certificate}. We use exactly the same circuit bank as in that analysis. For each of the ten GA seeds, the ten circuits with the lowest and highest inner-training prototype losses are retained.

For each circuit, we compute the RKHS-normalized geometric OvO margin $\widetilde m_i$ defined previously. Within each seed and loss group, we first average this margin over the ten circuits and then calculate the paired change for the same test sample,
\begin{equation*}
    \Delta \widetilde m_i
    =
    \widetilde m_i^{\mathrm{Low}}
    -
    \widetilde m_i^{\mathrm{High}}.
\end{equation*}
A positive value therefore means that the Low-loss circuits improve the sample's weakest OvO boundary.

To avoid defining sample difficulty using the same noisy score being compared, we use a leave-one-seed-out reference. For seed $s$, each test sample is ranked by its mean High-loss margin over the other nine seeds and assigned to one of ten equally sized difficulty groups. Decile 1 contains the hardest 10\% of samples. The Low-minus-High margin change for seed $s$ is then evaluated within these independently determined groups.

Figure~\ref{fig:app}(b) shows that the discrepancy is concentrated in the lower-margin tail rather than distributed uniformly across the test set. Across all paired test observations, the Low-loss circuits improve the geometric margin in $70\%$ of cases, increasing the mean margin from $0.0533$ to $0.0873$. Most difficulty deciles therefore benefit from the geometry favored by the prototype objective.

The behavior reverses, however, for the most difficult samples. Only $38\%$ of the observations in the hardest decile improve, and their mean margin change is $-0.0354$. Consistently, the test lower-decile margin decreases from $-0.0059$ to $-0.0614$, while the fraction of samples with a nonpositive weakest margin increases from $12.0\%$ to $19.2\%$. Thus, the increase in the average margin masks a deterioration near the most difficult OvO boundaries.

These results indicate an objective--downstream mismatch rather than a failure to optimize the prototype objective. Lower prototype loss improves the average prototype geometry and benefits most samples, but the mean-based objective does not explicitly protect the lower-margin tail that is particularly important to the QKSVM.


\begin{thebibliography}{38}
\providecommand{\natexlab}[1]{#1}
\providecommand{\url}[1]{\texttt{#1}}
\expandafter\ifx\csname urlstyle\endcsname\relax
  \providecommand{\doi}[1]{doi: #1}\else
  \providecommand{\doi}{doi: \begingroup \urlstyle{rm}\Url}\fi

\bibitem[Altares-López et~al.(2021)Altares-López, Ribeiro, and García-Ripoll]{IOP_NSGA2}
Sergio Altares-López, Angela Ribeiro, and Juan~José García-Ripoll.
\newblock Automatic design of quantum feature maps.
\newblock \emph{Quantum Science and Technology}, 6\penalty0 (4):\penalty0 045015, aug 2021.
\newblock \doi{10.1088/2058-9565/ac1ab1}.
\newblock URL \url{https://doi.org/10.1088/2058-9565/ac1ab1}.

\bibitem[Bae \& Kwek(2015)Bae and Kwek]{Bae_2015}
Joonwoo Bae and Leong-Chuan Kwek.
\newblock Quantum state discrimination and its applications.
\newblock \emph{Journal of Physics A: Mathematical and Theoretical}, 48\penalty0 (8):\penalty0 083001, jan 2015.
\newblock \doi{10.1088/1751-8113/48/8/083001}.
\newblock URL \url{https://doi.org/10.1088/1751-8113/48/8/083001}.

\bibitem[Banchi et~al.(2021)Banchi, Pereira, and Pirandola]{qml_generalization_banchi}
Leonardo Banchi, Jason Pereira, and Stefano Pirandola.
\newblock Generalization in quantum machine learning: A quantum information standpoint.
\newblock \emph{PRX Quantum}, 2:\penalty0 040321, Nov 2021.
\newblock \doi{10.1103/PRXQuantum.2.040321}.
\newblock URL \url{https://link.aps.org/doi/10.1103/PRXQuantum.2.040321}.

\bibitem[Barreto \& Neto(2005)Barreto and Neto]{vert}
Guilherme Barreto and Ajalmar Neto.
\newblock {Vertebral Column}.
\newblock UCI Machine Learning Repository, 2005.
\newblock {DOI}: https://doi.org/10.24432/C5K89B.

\bibitem[Biamonte et~al.(2017)Biamonte, Wittek, Pancotti, Rebentrost, Wiebe, and Lloyd]{qml_biamonte}
Jacob Biamonte, Peter Wittek, Nicola Pancotti, Patrick Rebentrost, Nathan Wiebe, and Seth Lloyd.
\newblock Quantum machine learning.
\newblock \emph{Nature}, 549\penalty0 (7671):\penalty0 195--202, 2017.

\bibitem[Caro et~al.(2021)Caro, Gil-Fuster, Meyer, Eisert, and Sweke]{PQC_encoding_gen}
Matthias~C. Caro, Elies Gil-Fuster, Johannes~Jakob Meyer, Jens Eisert, and Ryan Sweke.
\newblock Encoding-dependent generalization bounds for parametrized quantum circuits.
\newblock \emph{{Quantum}}, 5:\penalty0 582, November 2021.
\newblock ISSN 2521-327X.
\newblock \doi{10.22331/q-2021-11-17-582}.
\newblock URL \url{https://doi.org/10.22331/q-2021-11-17-582}.

\bibitem[Cortes et~al.(2012)Cortes, Mohri, and Rostamizadeh]{CKTA}
Corinna Cortes, Mehryar Mohri, and Afshin Rostamizadeh.
\newblock Algorithms for learning kernels based on centered alignment.
\newblock \emph{The Journal of Machine Learning Research}, 13:\penalty0 795--828, 2012.

\bibitem[Daniels et~al.(2022)Daniels, Wang, Simic, Bhargava, Capponi, Tonai, Yu, Bianco, and Lim]{tcell}
Kyle~G. Daniels, Shangying Wang, Milos~S. Simic, Hersh~K. Bhargava, Sara Capponi, Yurie Tonai, Wei Yu, Simone Bianco, and Wendell~A. Lim.
\newblock Decoding car t cell phenotype using combinatorial signaling motif libraries and machine learning.
\newblock \emph{Science}, 378\penalty0 (6625):\penalty0 1194--1200, 2022.
\newblock \doi{10.1126/science.abq0225}.
\newblock URL \url{https://www.science.org/doi/abs/10.1126/science.abq0225}.

\bibitem[Fisher(1936)]{iris}
R.~A. Fisher.
\newblock {Iris}.
\newblock UCI Machine Learning Repository, 1936.
\newblock {DOI}: https://doi.org/10.24432/C56C76.

\bibitem[German(1987)]{glass}
B.~German.
\newblock {Glass Identification}.
\newblock UCI Machine Learning Repository, 1987.
\newblock {DOI}: https://doi.org/10.24432/C5WW2P.

\bibitem[Glick et~al.(2024)Glick, Gujarati, Córcoles, Kim, Kandala, Gambetta, and Temme]{covariant_qkernel}
Jennifer~R. Glick, Tanvi~P. Gujarati, Antonio~D. Córcoles, Youngseok Kim, Abhinav Kandala, Jay~M. Gambetta, and Kristan Temme.
\newblock Covariant quantum kernels for data with group structure.
\newblock \emph{Nature Physics}, 20\penalty0 (3):\penalty0 479–483, January 2024.
\newblock ISSN 1745-2481.
\newblock \doi{10.1038/s41567-023-02340-9}.
\newblock URL \url{http://dx.doi.org/10.1038/s41567-023-02340-9}.

\bibitem[Havl{\'\i}{\v{c}}ek et~al.(2019)Havl{\'\i}{\v{c}}ek, C{\'o}rcoles, Temme, Harrow, Kandala, Chow, and Gambetta]{zzfeaturemap}
Vojt{\v{e}}ch Havl{\'\i}{\v{c}}ek, Antonio~D C{\'o}rcoles, Kristan Temme, Aram~W Harrow, Abhinav Kandala, Jerry~M Chow, and Jay~M Gambetta.
\newblock Supervised learning with quantum-enhanced feature spaces.
\newblock \emph{Nature}, 567\penalty0 (7747):\penalty0 209--212, 2019.

\bibitem[He et~al.(2016)He, Zhang, Ren, and Sun]{residual}
Kaiming He, Xiangyu Zhang, Shaoqing Ren, and Jian Sun.
\newblock Deep residual learning for image recognition.
\newblock In \emph{Proceedings of the IEEE conference on computer vision and pattern recognition}, pp.\  770--778, 2016.

\bibitem[Hearst et~al.(1998)Hearst, Dumais, Osuna, Platt, and Scholkopf]{SVM}
M.A. Hearst, S.T. Dumais, E.~Osuna, J.~Platt, and B.~Scholkopf.
\newblock Support vector machines.
\newblock \emph{IEEE Intelligent Systems and their Applications}, 13\penalty0 (4):\penalty0 18--28, 1998.
\newblock \doi{10.1109/5254.708428}.

\bibitem[Heo \& Park(2026)Heo and Park]{GQE}
Jaewoong Heo and Daniel~K. Park.
\newblock Generative quantum data embeddings for supervised learning.
\newblock \emph{arXiv preprint arXiv:2605.30866}, 2026.

\bibitem[Huang et~al.(2021)Huang, Broughton, Mohseni, Babbush, Boixo, Neven, and McClean]{powerofdata}
Hsin-Yuan Huang, Michael Broughton, Masoud Mohseni, Ryan Babbush, Sergio Boixo, Hartmut Neven, and Jarrod~R McClean.
\newblock Power of data in quantum machine learning.
\newblock \emph{Nature communications}, 12\penalty0 (1):\penalty0 2631, 2021.

\bibitem[Hubregtsen et~al.(2022)Hubregtsen, Wierichs, Gil-Fuster, Derks, Faehrmann, and Meyer]{TQE_physicalreview}
Thomas Hubregtsen, David Wierichs, Elies Gil-Fuster, Peter-Jan H.~S. Derks, Paul~K. Faehrmann, and Johannes~Jakob Meyer.
\newblock Training quantum embedding kernels on near-term quantum computers.
\newblock \emph{Phys. Rev. A}, 106:\penalty0 042431, Oct 2022.
\newblock \doi{10.1103/PhysRevA.106.042431}.
\newblock URL \url{https://link.aps.org/doi/10.1103/PhysRevA.106.042431}.

\bibitem[Hur \& Park(2025)Hur and Park]{quantum_margin}
Tak Hur and Daniel~K. Park.
\newblock Understanding generalization in quantum machine learning with margins.
\newblock In \emph{Proceedings of the 42nd International Conference on Machine Learning}, volume 267 of \emph{Proceedings of Machine Learning Research}, pp.\  26338--26360. PMLR, 13--19 Jul 2025.
\newblock URL \url{https://proceedings.mlr.press/v267/hur25a.html}.

\bibitem[Hur et~al.(2024)Hur, Araujo, and Park]{NQE}
Tak Hur, Israel~F. Araujo, and Daniel~K. Park.
\newblock Neural quantum embedding: Pushing the limits of quantum supervised learning.
\newblock \emph{Phys. Rev. A}, 110:\penalty0 022411, Aug 2024.
\newblock \doi{10.1103/PhysRevA.110.022411}.
\newblock URL \url{https://link.aps.org/doi/10.1103/PhysRevA.110.022411}.

\bibitem[Javadi-Abhari et~al.(2024)Javadi-Abhari, Treinish, Krsulich, Wood, Lishman, Gacon, Martiel, Nation, Bishop, Cross, et~al.]{qiskit}
Ali Javadi-Abhari, Matthew Treinish, Kevin Krsulich, Christopher~J Wood, Jake Lishman, Julien Gacon, Simon Martiel, Paul~D Nation, Lev~S Bishop, Andrew~W Cross, et~al.
\newblock Quantum computing with qiskit.
\newblock \emph{arXiv preprint arXiv:2405.08810}, 2024.

\bibitem[Katoch et~al.(2020)Katoch, Chauhan, and Chahar]{GA}
Sourabh Katoch, Sumit Chauhan, and Vijay Chahar.
\newblock A review on genetic algorithm: past, present, and future.
\newblock \emph{Multimedia Tools and Applications}, 80:\penalty0 8091--8126, Oct 2020.
\newblock \doi{10.1007/s11042-020-10139-6}.

\bibitem[K{\"u}bler et~al.(2021)K{\"u}bler, Buchholz, and Sch{\"o}lkopf]{qk_inductive}
Jonas K{\"u}bler, Simon Buchholz, and Bernhard Sch{\"o}lkopf.
\newblock The inductive bias of quantum kernels.
\newblock \emph{Advances in Neural Information Processing Systems}, 34:\penalty0 12661--12673, 2021.

\bibitem[Lecun et~al.(1998)Lecun, Bottou, Bengio, and Haffner]{mnist}
Y.~Lecun, L.~Bottou, Y.~Bengio, and P.~Haffner.
\newblock Gradient-based learning applied to document recognition.
\newblock \emph{Proceedings of the IEEE}, 86\penalty0 (11):\penalty0 2278--2324, 1998.
\newblock \doi{10.1109/5.726791}.

\bibitem[Lei et~al.(2024)Lei, Du, Mi, Yu, and Liu]{neuro_serach_iclr}
Cong Lei, Yuxuan Du, Peng Mi, Jun Yu, and Tongliang Liu.
\newblock Neural auto-designer for enhanced quantum kernels.
\newblock In \emph{International Conference on Learning Representations}, volume 2024, pp.\  17817--17841, 2024.

\bibitem[Lloyd et~al.(2020)Lloyd, Schuld, Ijaz, Izaac, and Killoran]{TQE}
Seth Lloyd, Maria Schuld, Aroosa Ijaz, Josh Izaac, and Nathan Killoran.
\newblock Quantum embeddings for machine learning.
\newblock \emph{arXiv preprint arXiv:2001.03622}, 2020.

\bibitem[Mandelbaum(2025)]{ibm_nighthawk}
Ryan Mandelbaum.
\newblock Scaling for quantum advantage and beyond.
\newblock IBM Quantum Computing Blog, November 2025.
\newblock URL \url{https://www.ibm.com/quantum/blog/qdc-2025}.

\bibitem[Moretti et~al.(2025)Moretti, Giachero, Radescu, and Grossi]{GA_moretti}
R~Moretti, A~Giachero, V~Radescu, and M~Grossi.
\newblock Enhanced feature encoding and classification on distributed quantum hardware.
\newblock \emph{Machine Learning: Science and Technology}, 6\penalty0 (1):\penalty0 015056, feb 2025.
\newblock \doi{10.1088/2632-2153/adb4bc}.
\newblock URL \url{https://doi.org/10.1088/2632-2153/adb4bc}.

\bibitem[Nguyen \& Chen(2022)Nguyen and Chen]{graph_search}
Nam Nguyen and Kwang-Cheng Chen.
\newblock Quantum embedding search for quantum machine learning.
\newblock \emph{IEEE Access}, 10:\penalty0 41444--41456, 2022.
\newblock \doi{10.1109/ACCESS.2022.3167398}.

\bibitem[Paulsen \& Raghupathi(2016)Paulsen and Raghupathi]{rkhs_paulsen}
Vern~I Paulsen and Mrinal Raghupathi.
\newblock \emph{An introduction to the theory of reproducing kernel Hilbert spaces}, volume 152.
\newblock Cambridge university press, 2016.

\bibitem[Pedregosa et~al.(2011)Pedregosa, Varoquaux, Gramfort, Michel, Thirion, Grisel, Blondel, Prettenhofer, Weiss, Dubourg, Vanderplas, Passos, Cournapeau, Brucher, Perrot, and Duchesnay]{scikit-learn}
F.~Pedregosa, G.~Varoquaux, A.~Gramfort, V.~Michel, B.~Thirion, O.~Grisel, M.~Blondel, P.~Prettenhofer, R.~Weiss, V.~Dubourg, J.~Vanderplas, A.~Passos, D.~Cournapeau, M.~Brucher, M.~Perrot, and E.~Duchesnay.
\newblock Scikit-learn: Machine learning in {P}ython.
\newblock \emph{Journal of Machine Learning Research}, 12:\penalty0 2825--2830, 2011.

\bibitem[Pérez-Salinas et~al.(2020)Pérez-Salinas, Cervera-Lierta, Gil-Fuster, and Latorre]{data_reupload}
Adrián Pérez-Salinas, Alba Cervera-Lierta, Elies Gil-Fuster, and José Latorre.
\newblock Data re-uploading for a universal quantum classifier.
\newblock \emph{Quantum}, 4:\penalty0 226, 02 2020.
\newblock \doi{10.22331/q-2020-02-06-226}.

\bibitem[Sahin et~al.(2024)Sahin, Symons, Pati, Minhas, Millar, Gabrani, Mensa, and Robertus]{sub_qkernel}
M.~Emre Sahin, Benjamin C.~B. Symons, Pushpak Pati, Fayyaz Minhas, Declan Millar, Maria Gabrani, Stefano Mensa, and Jan~Lukas Robertus.
\newblock Efficient {P}arameter {O}ptimisation for {Q}uantum {K}ernel {A}lignment: {A} {S}ub-sampling {A}pproach in {V}ariational {T}raining.
\newblock \emph{{Quantum}}, 8:\penalty0 1502, October 2024.
\newblock ISSN 2521-327X.
\newblock \doi{10.22331/q-2024-10-18-1502}.
\newblock URL \url{https://doi.org/10.22331/q-2024-10-18-1502}.

\bibitem[Schuld \& Killoran(2019)Schuld and Killoran]{q_featuremap}
Maria Schuld and Nathan Killoran.
\newblock Quantum machine learning in feature hilbert spaces.
\newblock \emph{Phys. Rev. Lett.}, 122:\penalty0 040504, Feb 2019.
\newblock \doi{10.1103/PhysRevLett.122.040504}.
\newblock URL \url{https://link.aps.org/doi/10.1103/PhysRevLett.122.040504}.

\bibitem[Schuld et~al.(2015)Schuld, Sinayskiy, and Petruccione]{qml_intro}
Maria Schuld, Ilya Sinayskiy, and Francesco Petruccione.
\newblock An introduction to quantum machine learning.
\newblock \emph{Contemporary Physics}, 56\penalty0 (2):\penalty0 172--185, 2015.
\newblock \doi{10.1080/00107514.2014.964942}.
\newblock URL \url{https://doi.org/10.1080/00107514.2014.964942}.

\bibitem[Schuld et~al.(2021)Schuld, Sweke, and Meyer]{encoding_expressive}
Maria Schuld, Ryan Sweke, and Johannes~Jakob Meyer.
\newblock Effect of data encoding on the expressive power of variational quantum-machine-learning models.
\newblock \emph{Phys. Rev. A}, 103:\penalty0 032430, Mar 2021.
\newblock \doi{10.1103/PhysRevA.103.032430}.
\newblock URL \url{https://link.aps.org/doi/10.1103/PhysRevA.103.032430}.

\bibitem[Shalev-Shwartz et~al.(2007)Shalev-Shwartz, Singer, and Srebro]{svm_obj}
Shai Shalev-Shwartz, Yoram Singer, and Nathan Srebro.
\newblock Pegasos: Primal estimated sub-gradient solver for svm.
\newblock In \emph{Proceedings of the 24th international conference on Machine learning}, pp.\  807--814, 2007.

\bibitem[Thanasilp et~al.(2023)Thanasilp, Wang, Nghiem, Coles, and Cerezo]{qml_trainability}
Supanut Thanasilp, Samson Wang, Nhat~Anh Nghiem, Patrick Coles, and Marco Cerezo.
\newblock Subtleties in the trainability of quantum machine learning models.
\newblock \emph{Quantum Machine Intelligence}, 5\penalty0 (1):\penalty0 21, 2023.

\bibitem[Tolstikhin et~al.(2018)Tolstikhin, Bousquet, Gelly, and Schoelkopf]{wae}
Ilya Tolstikhin, Olivier Bousquet, Sylvain Gelly, and Bernhard Schoelkopf.
\newblock Wasserstein auto-encoders.
\newblock In \emph{International Conference on Learning Representations}, 2018.
\newblock URL \url{https://openreview.net/forum?id=HkL7n1-0b}.

\end{thebibliography}
\end{document}